%% file: main.tex
\documentclass[a4paper,UKenglish, autoref,cleveref, thm-restate]{lipics-v2021}

\newif\iflong
\newif\ifshort
\newif\ifdraft%
\draftfalse

\iflong
\else
\shorttrue
\fi

\input{preamble}

\title{Hyperbolic Random Graphs: Clique Number and Degeneracy with Implications for Colouring} 

\titlerunning{Hyperbolic Random Graphs:
Clique Number and Degeneracy} 

 \author{Samuel Baguley}{Hasso Plattner Institute, University of Potsdam, Germany}{Samuel.Baguley@hpi.de}{}{}

\author{Yannic Maus}{TU Graz, Austria}{yannic.maus@tugraz.at}{}{}

\author{Janosch Ruff}{Hasso Plattner Institute, University of Potsdam, Germany}{Janosch.Ruff@hpi.de}{}{}

\author{George Skretas}{Hasso Plattner Institute, University of Potsdam, Germany}{Georgios.Skretas@hpi.de}{}{}

\authorrunning{S.~Baguley, Y.~Maus, J.~Ruff, G.~Skretas} 

\Copyright{Samuel~Baguley, Yannic.~Maus, Janosch~Ruff, George~Skretas} 



\category{} 

\relatedversion{An extended abstract of this paper appeared at STACS 2025.} 

\funding{This research was partially funded by the German Research Foundation (Deutsche
Forschungsgemeinschaft, DFG) – project number 390859508, and by Austrian Science Fund (FWF) \url{https://doi.org/10.55776/P36280},  \url{https://doi.org/10.55776/I6915}, and \url{https://doi.org/10.55776/DOC183}.}

\acknowledgements{The authors would like to thank Thomas Bläsius for enriching discussions and Maximilian Katzmann for running experiments indicating a gap between core and maximum \inner during a research visit of JR in Karlsruhe.}

\nolinenumbers 

\hideLIPIcs  

\EventEditors{}
\EventLongTitle{42nd International Symposium on Theoretical Aspects of Computer Science (STACS 2025)}
\EventShortTitle{STACS 2025}
\EventAcronym{STACS}
\EventYear{2025}
\EventDate{March 4--7, 2025}
\EventLocation{Jena, Germany}
\EventLogo{}
\SeriesVolume{}
\ArticleNo{13}

\begin{document}

\maketitle
	
\begin{abstract}

Hyperbolic random graphs inherit many properties that are present in real-world networks. The hyperbolic geometry imposes a scale-free network  with a strong clustering coefficient. Other properties like a giant component, the small world phenomena and others follow. This motivates the design of simple algorithms for hyperbolic random graphs. 

In this paper we consider threshold hyperbolic random graphs (HRGs). 
Greedy heuristics are commonly used in practice as they deliver a good approximations to the optimal solution even though their theoretical analysis would suggest otherwise. A typical example for HRGs are degeneracy-based greedy algorithms [Bläsius, Fischbeck; Transactions of Algorithms '24].
In an attempt to bridge this theory-practice gap we characterise the parameter of degeneracy yielding a simple approximation algorithm for colouring HRGs. The approximation ratio of our algorithm ranges from $(2/\sqrt{3})$ to $4/3$ depending on the power-law exponent of the model.    We complement our findings for the degeneracy with new insights on the clique number of hyperbolic random graphs. We show that degeneracy and clique number are substantially different and derive an improved  upper bound on the clique number. Additionally, we show that the core of HRGs does not constitute the largest clique. 

Lastly we demonstrate that the degeneracy of the closely related standard model of geometric inhomogeneous random graphs behaves inherently different compared to the one of hyperbolic random graphs.
\end{abstract}

\newpage

\section{Introduction}

Many real-world networks have a heterogeneous degree distribution, close to a power-law, as well as a constant clustering coefficient. The \emph{hyperbolic random graph} model (HRG) introduced by Krioukov~et.~al.~\cite{PhysRevE.82.036106} combines both properties \cite{gpp-hrg-12}, which has led to considerable interest in recent years. Various aspects of HRGs have been studied, including clique size \cite{bfm-cliques-17, Stegehuis-cliques-2023, Schiller2024}, treewidth \cite{bfk-tw-2016}, minimum vertex cover size \cite{katzmann-exactvc-2023, katzmann-approxvc-2023}, and diameter \cite{fk-dhrg-18, KM-diameter-15, ms-k-19}. A hyperbolic random graph is a graph embedded in the hyperbolic plane where pairs of vertices have an edge if 
they are close according to the hyperbolic distance. It is generated by randomly throwing $n$ vertices on a disk of radius $R$ (dependent on $n$). Since hyperbolic space grows exponentially, most vertices are of small degree and located close to the boundary of the disk, while few vertices have large degree and lie close to the centre. 
This distribution of vertices leads to the power-law degree distribution of HRGs.

In this paper, we study the vertex colouring problem on HRGs, along with the related concepts   of  clique number and degeneracy. The $k$-colouring problem asks to colour the vertices of a graph with $k$ colours, while assigning different colours to adjacent vertices. The \emph{chromatic number} \chrom is the minimum number of colours needed to colour a graph in such a way. For $k\geq 3$ the problem is one of the original NP-hard problems \cite{stoc-cook71, Karp1972}. On general graphs, even loosely approximating the chromatic number is particularly hard \cite{z-colourhardness-07}. 

The answer to the colouring problem is closely related to the clique number \clique and the degeneracy of a graph. The \emph{clique number} is the number of vertices of the largest clique of the graph and it serves as a natural lower bound to the chromatic number. The \emph{degeneracy \degen} is the minimum integer $k'$ for which there exists an ordering of the vertex set of $G$, $V = (v_1, v_2, \cdots, v_n)$, such that for every index $i \in [n-1]$, $v_i$ has at most $k'$ neighbours with greater index.  Any graph $G$ can be \emph{easily} coloured with $\degen+1$ colours by iterating through the vertices in  reverse order and simply colouring a vertex with a colour not used by any of its higher ranked neighbours.

We aim to study these structural parameters that are not only fundamental to the model but also in general for algorithm design in various models of computation \cite{Barenboim2013,CNR23,GG24}. The most prominent large clique in an HRG is formed by the vertices in the graph's core \cite{bfm-cliques-17}.
Simply put, the core emerges among polynomially many vertices of distance at most $R/2$ from the centre of the disk, which due to the triangle inequality form a clique. We denote the size of this clique by \coresize. 
At this point one may wonder whether the core forms the largest clique of the graph, a statement that we disprove in \Cref{prop:larger-cliques}.
Nevertheless we show that the largest clique can at most be small constant factor larger than the core. 
\begin{theorem*}[Simplified version of \cref{the:clique-upper-bound}]
\label{thm:informalCliqueUpper}
There exists a constant $\delta>0$ such that for any threshold HRG $G$,  
$\coresize+1\leq \clique\leq \sqrt{4/3-\delta}\cdot \coresize$ holds \whp\footnote{An event holds \emph{with extremely high probability} (\whp), if for every $c > 1$, there exists an $n_0$ such that for every $n \geq n_0$ the event holds with probability at least $1 - n^{-c}$.}
\end{theorem*}
This upper bound improves on prior work \cite{bfm-cliques-17}, which showed that there exists some constant $c>1$ such that $\clique\leq c\cdot\coresize$ holds \whp, but without providing any upper bound on $c$.

We can now see that the core and clique are not the same. Nonetheless, (i) this theorem shows that the largest clique size and the core size are closely related, and (ii) the core of HRGs is a very well understood object, whereas the largest clique is not. Thus the natural approach to bound the degeneracy is to use the core. We show the following theorem.

\begin{theorem*}[Simplified version of \cref{the:degeneracy-upper,the:degeneracy-upper}]
\label{thm:informalDegen} 
There exist constants $\delta_1,\delta_2>0$ such that for any threshold HRG $G$,  $(1+\delta_1)\cdot \coresize\leq \degen\leq (4/3-\delta_2)\cdot \coresize$ holds \whp
\end{theorem*}
The main surprise of this theorem is that the degeneracy is bounded away from the coresize by a constant factor. 
As the chromatic number is  lower bounded by the core size, the upper bound on the degeneracy in this theorem implies a simple algorithm colouring with at most $(4/3-\delta_2)\chrom$ colours. The approximation guarantee of this algorithm ranges from $2/\sqrt{3}$ to $4/3$ depending on further model parameters, see \Cref{sec:preliminaries} and \Cref{the:approx-algo} for details. In any case, this improves on the previously best approximation ratio of $2$  \cite[Lemma 7]{bfks-unitdisk-23}. 
The algorithm iteratively removes vertices of degree at most $(4/3-\delta_2)\coresize-1$ and then colours them in the reverse order. 
The next thing one would hope is to be able colour the graph with \clique colours using the same process.  In \cite{bf-externalval-2024}, the authors conducted experiments where they were iteratively removing the vertex with the smallest degree of the graph, up to vertices with residual degree equal to \clique. In their findings, this process did not remove every vertex, implying \clique< \degen for their generated graphs. We substantiate their findings by providing a rigorous proof demonstrating that the clique number is, in fact, a constant factor smaller than the degeneracy. 

\begin{theorem*}[Simplified version of \cref{the:clique-deg-gap}]
\label{thm:informalClique}
There exists a constant $\eps>0$ such that for any threshold HRG $G$,  
$\clique\leq (1-\eps)\cdot\degen$ holds \whp 
\end{theorem*}

Our final contribution is to study the degeneracy  of \emph{geometric inhomogeneous random graphs} (GIRGs) \cite{k-girg-18}, a sibling to HRGs. The GIRGs also combine heterogeneity and high clustering. For most properties GIRGs and HRGs exhibit the same behaviour. Perhaps the first paper to find a difference between them is \cite{Stegehuis-cliques-2023}, where the authors show that the minimum number of maximal cliques in the two models differ. We show a significant discrepancy for the degeneracy of GIRGs compared to that of HRGs, see \Cref{fig:plots} and \Cref{cor:hrg-girg-gap}.

\medskip
\noindent\textbf{Outline.} See \cref{fig:plots} for a table with our results, as well as a plot comparing the bounds of our theorems for various model parameters. In \cref{sec:summary} we provide a detailed discussion of our results and techniques. \cref{sec:inner-neighbourhood} contains bounds on the degeneracy of HRGs (\Cref{the:degeneracy-upper}). In \cref{sec:clique-number}, we show the gap between clique number and degeneracy (\Cref{the:clique-deg-gap}), as well as bounds on the clique number (\Cref{the:clique-upper-bound}). Finally, \cref{sec:girg} contains results about the degeneracy of GIRGs. Statements where proofs or technical details are omitted can be found in the~\cref{appendix:start}.

\begin{figure}
     \centering
     \begin{subfigure}[b]{0.55\textwidth}
         \centering
         \includegraphics[width=\textwidth]{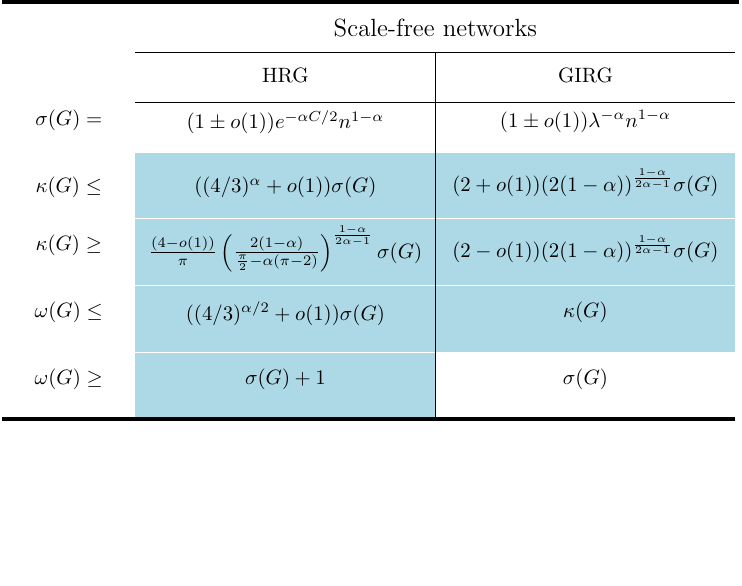}
         \label{fig:table}
     \end{subfigure}
     \hfill
     \begin{subfigure}[b]{0.44\textwidth}
         \centering
         \includegraphics[width=\textwidth]{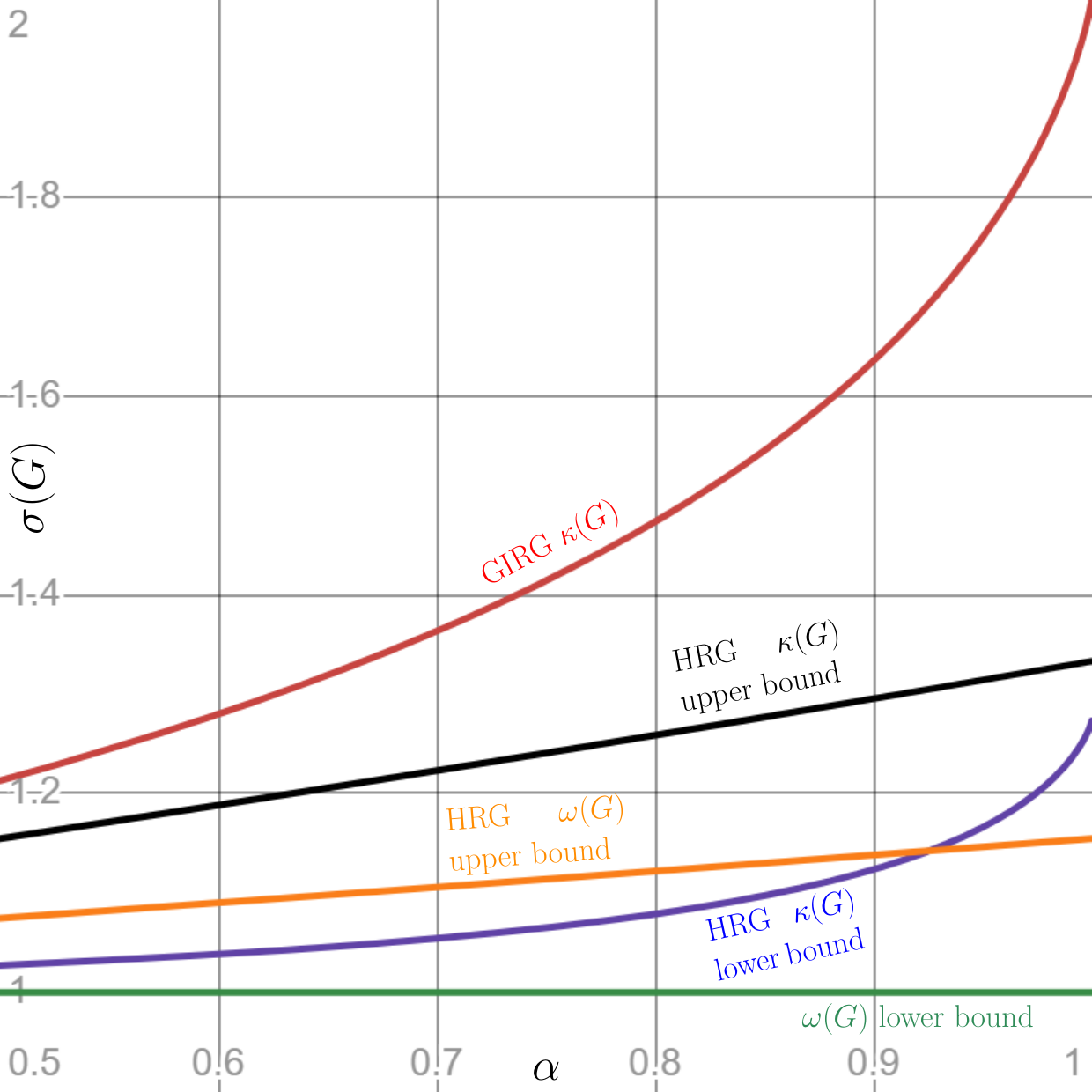}
         \label{fig:plot}
     \end{subfigure}
        \caption{Results on the degeneracy $\degen$ and the size of the largest clique \clique in hyperbolic random graphs (HRG) and geometric inhomogeneous random graphs (GIRG). The bounds hold \whp and are stated in comparison to the core size \coresize. Each curve represents the multiplicative factor in front of $\coresize$ for $\degen$ and $\clique$ (y-axis) depending on the parameter $\alpha \in (1/2, 1)$ (x-axis). Prior work is listed with white background, whereas our results are listed with blue background.}
        \label{fig:plots}
\end{figure}

\subsection{Discussion of our Results and Techniques}\label{sec:summary}
HRGs have a power-law degree distribution \cite{gpp-hrg-12, Papadopoulos2010}, that is, the probability that a vertex has degree $k$ is given by $\sim k^{-(2\alpha + 1)}$ . The model parameter $\alpha\in (1/2,1)$ controls the power-law exponent and all our results, particularly the size of the aforementioned constant factor gap depends on the choice of $\alpha$. For the ease of presentation this overview largely omits this dependence, but the summary of our results in \Cref{fig:plots} plots it in detail. 

\smallskip

\noindent\textbf{Upper bound on degeneracy (\Cref{the:degeneracy-upper}).} One consequence of generating a graph in hyperbolic space is that vertices tend to have fewer neighbours with increasing radius, i.e., the expected number of neighbours of a vertex decreases with the distance from the centre of the hyperbolic disc. This produces the power-law degree distribution that is valuable in modelling real-world networks. It also leads to a simple approach for upper bounding degeneracy: instead of removing vertices ordered by (increasing) degree, we remove them by (decreasing) radius. 
If $k$ is such that each vertex has at most $k$ neighbours of smaller radius, then $k$ is an upper bound on the degeneracy.

The notion governing this approach is the \emph{\inner} of a vertex, see \cref{fig:inner-neighbourhood}. The \inner of a vertex $u$ with radius $r$, denoted $\innerdeg(u)$, is the set of vertices of distance at most $R$ from $u$, i.e., they are neighbours of $u$, and with radius at most $r$, i.e., they are closer to the centre of the disc than $u$. The size $|\innerdeg(u)|$ of the \inner is called the \emph{\indeg} of $u$. The expected value of $|\innerdeg(u)|$ scales with the area of $u$'s \emph{inner-ball}  $\mathcal I(r) = \B_u(R)\cap \B_0(r)$. More precisely, 
$\E{|\innerdeg(u)|}=(n-1)\cdot\mu(\mathcal I(r))$. See \cref{fig:inner-neighbourhood}~a for a visual representation. The vertex of largest \indeg is denoted $\ustar$; the choice of $\ustar$ and the value of $|\innerdeg(\ustar)|$ depend on the random distribution of the vertices.

\begin{figure}[t]
    \centering \includegraphics[height=0.24\textheight]{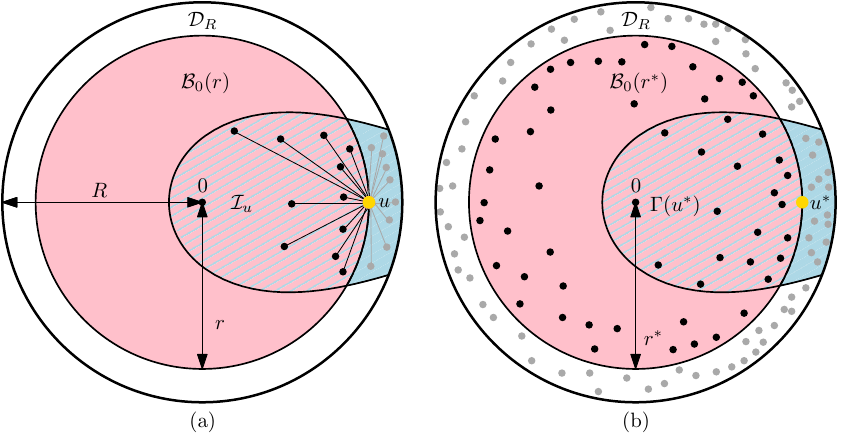}
    \caption{Illustration of the inner-neighbourhood. (a) The pink area is the ball $\mathcal{B}_0(r)$. The hatched area $\mathcal{I}_u = \mathcal{B}_u(R) \cap \mathcal{B}_0(r)$ is the inner-ball of $u$. The vertices $V \cap \mathcal{I}_u$ form the \inner $\innerdeg(u)$.
    (b) Sketch of proof of \cref{lem:deg-lower-inner}. Vertex $u^*$ is the vertex with the largest inner-degree. Each vertex $v \in U = V \cap \mathcal{B}_0(r^*)$ has at least nearly as many neighbours within $\mathcal{B}_0(r^*)$ as $|\innerdeg(\ustar)|$ (\whp)}
    \label{fig:inner-neighbourhood}
\end{figure}

If vertices are removed from outside to inside, then each vertex at the time of its removal will have degree less or equal to the \indeg of $\ustar$.
To bound the degeneracy, we derive a probabilistic upper bound for $|\innerdeg(\ustar)|$. We do this by finding the radius that maximises $\mu(\mathcal I(r))$, which upper bounds $\E{|\innerdeg(u)|}$ (for every vertex $u$). 
Since $|\innerdeg(u)|$ is concentrated we can apply a Chernoff and a union bound to obtain a high probability upper bound for $|\innerdeg(\ustar)|$.  Despite the simplicity of the \inner, we elaborate on this key concept as it is not only crucial to our upper bound on degeneracy but also for most of our other results discussed later on.

\smallskip
\noindent\textbf{Lower bound on degeneracy (\Cref{the:degeneracy-upper}).} In \cref{lem:deg-lower-inner}  we show that the maximum \indeg also produces a lower bound on the degeneracy \degen, in the sense that $\degen\geq  (1-o(1))|\innerdeg(u^*)|$ \whp This yields asymptotically tight bounds on \degen. We prove the lower bound of \degen by considering the subgraph $G'$ induced by the vertices that have smaller radius than \ustar (see \cref{fig:inner-neighbourhood}~b). Note that $G'$ contains vertices that are not neighbours of $u^*$. We show that every vertex $u$ in $G'$ has, \whp, at least $(1-o(1))|\innerdeg(u^*)|$ neighbours in $G'$; this statement uses the choice of $u^*$ and is not true if $u$ was an arbitrary vertex. Now, in any ordering of the vertices of $G'$, 
the first vertex has at least $(1-o(1))|\innerdeg(u^*)|$ neighbours of greater index, implying  $(1-o(1))|\innerdeg(u^*)|\le\kappa(G')\le\degen$. 
While this provides a lower bound that asymptotically matches our upper bound, a lower order gap remains.

\smallskip
\noindent\textbf{Gap between clique number and degeneracy (\Cref{the:clique-deg-gap}).} The most immediate lower bound for degeneracy is the clique number \cite{Barenboim2013}, because in every ordering of the vertices of $G$, the vertex of the clique with the lowest index has at least $\clique-1$ neighbours of higher index.
We prove that the lower bounds on the degeneracy that are derived from the clique number are strictly worse than the bounds discussed above obtained via analysing the inner-degree. Our approach to show this is the following: we take an arbitrary clique $K$ and the vertex $u$ with the largest radius in $K$. Let $G''$ be the subgraph induced by $\innerdeg(u)$. Next, we partition $G''$ intro three sets of vertices, each containing at least a constant fraction of the vertices of $G''$ \whp See \Cref{fig:separator}~b for an illustration of this partition. Lastly, we use purely geometric arguments to show that $K$ cannot contain vertices from all three sets. The gap then follows as the left out set contains a constant fraction of $u$'s inner neighbourhood. 
This gap between \clique and \degen makes approaches for computing the clique number via degeneracy, like that of Walteros and Buchanan \cite{Walteros2020}, unsuitable for HRGs.

\smallskip
\noindent\textbf{Clique number and the core (\Cref{the:clique-upper-bound}).}
The upper bound on the clique number is proven via a geometric approach. 
Let $u, v, w$ be any three vertices. If the vertices are far apart they cannot be contained in a clique. Otherwise their pairwise distance is at most $R$. Using the hyperbolic of version Jung's theorem \cite{Jung1901,Dekster1995,Dekster1997} implies that they are contained in a ball $\B$ of small radius and we show that $\B$ also has a small area. Hence the expected number of vertices in $\B$ is small as well. As this expectation is well concentrated whenever the area is significantly large, this bound holds with extremely high probability when introducing lower order deviations from the expectation.  Via a union bound over all possible $\binom{n}{3}$ triples of vertices we rule out that any clique contained in such a covering ball is large. The main claim now follows because any clique has to be contained in one of these coverings balls. 
The question of whether $\clique/\coresize\to 1$ as $n\to\infty$ remains open.

\section{Preliminaries}\label{sec:preliminaries}

\subparagraph{Hyperbolic Random Graphs.}
We follow the formalisation of hyperbolic random graphs introduced in \cite{Papadopoulos2010}, which is known as the \emph{native representation}.
We denote by $\bH = [0,\infty)\times[0,2\pi)$ the hyperbolic plane in the polar coordinate system, where a point $x\in \bH$ is parametrised by a radius $r(x)$ and an angle $\varphi(x)$.
We equip \bH with a metric $\dist(x,y)$ characterised by 
\begin{align}\label{eq:hyperbolic_distance}
    \cosh(\dist(x,y)) = \cosh(r(x))\cosh(r(y)) - \sinh(r(y))\sinh(r(x))\cos(\varphi(x) - \varphi(y)).
\end{align}
This metric is what gives \bH a hyperbolic geometry, of curvature -1, as opposed to the Euclidean metric. We equip \bH with the topology induced by $\dist$.

The geometric space of most importance in this work is the bounded disk in \bH defined by $\disk = [0,R]\times[0,2\pi)$, where $R = 2\log(n) + C$ with $C \in \Theta(1)$.
We refer to point $(0,0)$ as the \emph{centre of this disk}. 
The space \disk inherits the topology of \bH, and from now on we shall only consider subsets of this space -- thus, for example, a ball around a point $x\in\disk$ is defined by the set $\B_x(\eps) = \{y\in\disk : \dist(x,y)\leq\eps\}\subseteq \disk$.

We now introduce a probability measure $\mu$ on \disk, which is parametrised by the model parameter $\alpha\in(1/2,1)$, and was first defined by Papadopoulos et.~al.~\cite{Papadopoulos2010}.
For measurable $\mathcal S\subseteq \disk$, define
\begin{align*}
	\mu(\mathcal S) = \int_S \rho(x) \dd x, 
    \qquad \rho(x) = \frac{\alpha\sinh(\alpha x)}{2\pi(\cosh(\alpha R) - 1)}, 
\end{align*}
where $\rho$ is the density of $\mu$ with respect to the Lebesgue measure on \disk. This measure differs from the uniform probability measure on \disk in that it puts more mass at the centre of the disk; both measures coincide at $\alpha=1$. The benefit of $\mu$ lies in the properties it induces in our central object of study, the hyperbolic random graph.

\subparagraph{Threshold hyperbolic random graph (HRG).} A \emph{(threshold) hyperbolic random graph} or \emph{HRG} is a pair $G = (V,E)$ defined by the following procedure. First, $n$ vertices are sampled independently at random in \disk according to $\mu$. Then any two vertices $u,v\in V$ are connected by an edge if and only if their distance $\dist(u,v)$ is at most $R$. We write $G \sim \G$ to denote a graph generated in this way. A vertex $u \in V$ is identified by its point coordinates in $\bH$ and we write $V \cap \mathcal{A}$ to denote the set of vertices that are located in an area $\mathcal{A} \subseteq \disk$.

The use of $\mu$ to distribute vertices in \disk has the effect of giving $G$ a power-law degree distribution, as was shown in \cite{gpp-hrg-12, Papadopoulos2010}. It is sometimes convenient to characterise connection of vertices in terms of their \emph{angular distance}, and to that end we define
\begin{align*}
    \theta_R(r_1,r_2) = \arccos\left(\frac{\cosh(r_1)\cosh(r_2) - \cosh(R)}
    {\sinh(r_1\sinh(r_2)}\right),
\end{align*}
which per \eqref{eq:hyperbolic_distance} yields the following observation.
\begin{observation}
Two vertices $u$ and $v$ are connected if and only if their angular distance is less than $\theta_R(r(u),r(v))$.
\end{observation}
We also make use of the following expression of the distribution of the radius of a vertex $u$.
\begin{align}\label{eq:measure-of-origin-ball}
    \Prob{r(u) \leq r} = \mu\left(\mathcal{B}_0(r)\right) = \int_0^r\int_{-\pi}^{\pi} \rho(x) \dd\theta \dd x = \frac{\cosh(\alpha r) - 1}{\cosh(\alpha R) - 1} = (1-o(1))e^{-\alpha(R-r)}
\end{align}

We briefly note that a variant of HRGs exists in which vertices are not connected purely according to whether their distance is below a threshold, but rather with probability $p(u,v) = (1 + \exp\left(\frac{1}{2T} \left(\dist(u, v) - R\right)\right))^{-1}$ determined by both distance and a ``temperature'' parameter $T$ (see e.g. \cite[§3.1]{Krohmer2016}).

\smallskip
\noindent\textbf{Degeneracy, clique number, chromatic number and core.}
For a graph $G=(V,E)$, the \emph{degeneracy} $\kappa(G)$ is the minimum integer
$k$ for which there exists an ordering of the vertex set of $G$, $V = (v_1, v_2, \cdots, v_n)$, such that for every index
$i \in [n-1]$, $v_i$ has at most $k$ neighbours with greater index. 
The \emph{clique number} $\omega(G)$ is the size of the largest clique of $G$. The \emph{chromatic number} $\chi(G)$ is the smallest number of colours required so that a conflict-free vertex colouring is possible for $G$. 
The \emph{core} of a hyperbolic random graph is the set of vertices with radius at most $R/2$ and we denote its size by $\coresize$. Since for any points $u,v \in \mathcal{B}_0(R/2)$ the distance is at most $\dist(u,v) \leq R$, the core forms a clique. Finally, since the core is a clique, any vertex of a clique needs a different colour in a conflict-free colouring, and $\chrom \leq \degen + 1$ (see e.g. \cite[Lemma 4]{Matula1983}), we have the following chain of inequalities.
\begin{observation}\label{lem:lower-clique-number}
Let $G \sim \mathcal{G}(n, \alpha, C)$ be a threshold HRG. Then, $$\coresize\leq \omega(G)\leq\chi(G)\leq\kappa(G)+1.$$
\end{observation}

\smallskip
\noindent\textbf{Concentration bounds.}
We use the following Chernoff bounds~\cite[Theorem 4.4]{mu-pc-05}.

\begin{theorem}[Chernoff bound]\label{lem:Chernoff}
For $i \in [k]$, let $X_i \in \{0,1\}$ be independent random variables and $X = \sum_i X_i$. Then for $\eps \in (0,1)$,
\begin{align*}
    \Prob{X\geq (1+\eps)\mathbb{E}[X]} \leq e^{-\eps^2 \cdot \mathbb{E}[X]/3} \text{ and }
    \Prob{X\leq (1-\eps)\mathbb{E}[X]} \leq e^{-\eps^2 \cdot \mathbb{E}[X] /2}.
\end{align*}
\end{theorem}

\section{Degeneracy of Hyperbolic Random Graphs}\label{sec:inner-neighbourhood}

A tool we make use of several times is the \emph{inner-ball} of a point $x \in \disk$, defined by $\mathcal{I}_x=\B_x(R)\cap\B_0(r(x))$. 
The inner-ball of a vertex is the inner-ball of the point using the vertex' coordinates. 
The \emph{\inner} of a vertex $u$ is the set of vertices (excluding $u$) contained in its inner-ball, that is, the neighbours of $u$ that have a smaller radius than $u$ (see \cref{fig:inner-neighbourhood}~a), and is denoted $\innerdeg(u)$. 

We upper bound the degeneracy \degen via the \inner by using the following informal process. Consider a graph $G$ and order its vertices $(v_1, v_2, \cdots, v_n)$ by decreasing radius, so $r(v_i) \geq r(v_{i+1})$, and iteratively remove vertices from $G$ one-by-one, from lower to higher index.
Note that the set of neighbours of $v_i$ that have greater index than $i$ coincide with $\innerdeg(v_i)$. This implies that the degree of each vertex $v_i$ at the time of its removal is $|\innerdeg(v_i)|$. Let $u^*$ be the vertex $v_k$ that maximises $|\innerdeg(v_k)|$. As the largest degree of a vertex at the time of its removal is given by $|\innerdeg(\ustar)|$, we get the following upper bound for the degeneracy.

\begin{observation}\label{obs:deg-upper-inner}
Let $G\sim\mathcal{G}(n, \alpha, C)$ be a threshold HRG and let $u^*$ be the vertex of $G$ with the largest inner-degree in $G$. Then $\kappa(G) \leq |\innerdeg(u^*)|$.  
\end{observation}

We will now show that the largest \indeg does not only yield an immediate upper bound on the degeneracy, but also a lower bound. Informally, this lower bound follows from the following argument. Order the vertices of the graph in descending order of their radius, that is, $(v_1,v_2,\ldots,v_n)$ such that $r(v_i) \geq r(v_{i+1})$. For $i \in [n]$, let $G_i=G\setminus\{v_1,v_2,\ldots, v_i\}$.
Let $v_k$ be the node with the maximum inner neighbourhood in $G$. We show (with a probabilistic guarantee) that the graph $G_k$ has minimum degree $(1-o(1))|\innerdeg(v_k)|$. Before we make this formal, we introduce a slightly modified version of \cite[Lemma 3.3]{gpp-hrg-12}, that implies the following: the closer a vertex is to the origin, the more neighbours it has in expectation.  
This can be derived via the fact that the angle $\theta_R(r,x)$ is monotonically decreasing in $x$.  

\begin{corollary}\label{lem:ball-lower-bound}

    Let $r, s, t \in[0,R)$ with $s < t$. Then
    $\mu(\mathcal{B}_s(R) \cap \mathcal{B}_0(r)) \geq \mu(\mathcal{B}_t(R) \cap \mathcal{B}_0(r))$.
\end{corollary}

\cref{lem:ball-lower-bound} tells us  that any vertex with radius at most $r$ has in expectation at least as many neighbours up to radius $r$, as the expected inner-degree of a vertex with radius exactly $r$.  In order to derive a high-probability bound on $|\innerdeg(\ustar)|$
it now suffices to show concentration around the expectation of all considered neighbourhoods.

Since the size of every clique $K$ is upper bounded by the \indeg of the vertex in $u \in K$ that has the largest radius, \clique is a lower bound for the maximum \indeg. We can now lower bound the degeneracy \degen based on the largest inner-degree. Note that, in contrast to the upper bound of \cref{obs:deg-upper-inner}, the lower bound is not deterministic.

\begin{lemma}\label{lem:deg-lower-inner}
Let $G \sim \mathcal{G}(n, \alpha, C)$ be a threshold HRG. Then $\degen \ge (1-o(1))|\innerdeg(u^*)|$ \whp
\end{lemma}
\begin{proof}
For any subgraph $H\subseteq G$, it is clear that $\degen \ge \min_{v\in V(H)} \mathrm{deg}|_H (v)$, where $\mathrm{deg}|_H (v) = \sum_{w\in V \setminus \{v\}}\indicator\{\{v,w\} \in E(H)\}$. We let $H$ be the (random) subgraph of $G$ created by restricting $G$ to vertices that land in $\B_0(r)$, that is, that have radius at most $r$, and keeping the same edges. Then for any $v\in V$, 
\begin{align*}
	\Exp[\mathrm{deg}|_H (v) | r(v)]
	&=\sum_{w\in V \setminus \{v\}} \P(w \in \B_0(r)\cap \B_{r(v)}(R) \given r(v))\indicator\{r(v)\le r\}\\
	&= (n-1) \mu(\B_0(r)\cap \B_{r(v)}(R))\indicator\{r(v)\le r\}\\
    &\ge (n-1) \mu(\B_0(r)\cap \B_r(R))\indicator\{r(v)\le r\}. \tag{\cref{lem:ball-lower-bound}}
\end{align*}

We write $\gamma \coloneqq (n-1) \mu(\B_0(r)\cap \B_r(R))$.
Since $\mathrm{deg}|_H (v)$ is a sum of Bernoulli random variables that are all independent under $\P( \,\cdot\, |\, r(v))$, a Chernoff bound (\cref{lem:Chernoff}) gives 
\begin{align*}
	\P\Big(\min_{v\in V(H)} \mathrm{deg}|_H (v) &< (1-\varepsilon)\gamma\Big)
	= \P\Big(\bigcup_{v\in V} \{\mathrm{deg}|_H (v) < (1-\varepsilon)\gamma, r(v) \le r\}\Big)\\
	&\le n \P(\mathrm{deg}|_H (v) < (1-\varepsilon)\gamma, r(v) \le r) \tag{Union bound}\\
	&= n \E{\P(\mathrm{deg}|_H (v) < (1-\varepsilon)\gamma\given r(v))\indicator\{r(v)\le r\}} \tag{Tower rule}\\
	&\le n \Exp\Big[ \ee^{-\eps^2\Exp[\mathrm{deg}|_H (v) | r(v)]/2}\indicator\{r(v)\le r\}\Big] \tag{Chernoff bound}\\
    &\le n \ee^{-\gamma\varepsilon^2/2}\P(r(v)\le r)
    \le n \ee^{-\gamma\varepsilon^2/2}.
\end{align*}
Taking $r$ to be the argmax of $\gamma$ yields $\gamma(r) \ge \Exp[\Gamma(\ustar)]$. By \cref{lem:lower-clique-number} and applying \eqref{eq:measure-of-origin-ball} at $R/2$ we have that
$
   \gamma(r) \ge \Exp[\Gamma(\ustar)] \ge \Exp[\coresize] \in n^{\Omega(1)}.
$
Thus, choosing $\varepsilon = 1/\log(n)$, we obtain for any constant $c$ as long as $n$ is large enough
\begin{align*}
    \P(\degen< (1-\varepsilon)|\innerdeg(u^*)|) 
    \le \P(\degen<(1-\varepsilon)\gamma(r))
    \le n \ee^{-\gamma(r)\varepsilon^2/2}
    \le n \ee^{-n^{\Omega(1)}/\log(n)}
    \le n^{1-c},
\end{align*}
that is, $\degen \ge (1-o(1))|\innerdeg(u^*)|$ \whp
\end{proof}

In the rest of this section we derive bounds for the largest \indeg on HRGs that hold \whp which, by \cref{obs:deg-upper-inner,lem:deg-lower-inner}, translate into results for the degeneracy. Since a vertex $v$ belongs to the \inner of a vertex $u$, if and only if it resides in the inner-ball of $u$, we can use the measure of an inner-ball $\mathcal{I}_u$ to bound the maximum \indeg of a graph. Moreover, the measure of the inner-ball is invariant under rotation around the origin, which is why we write $\mu\left(\mathcal{I}(r)\right)$ instead of $\mu\left(\mathcal{I}_u\right)$ for $r = r(u)$. We sum up our results for the area of the inner-ball in the following technical lemma.

\begin{lemma}[Volume of the inner-ball]\label{cor:inner-neighbourhood-equal}
    Let $\Delta \in \Theta(1)$ and let $u \in \disk$ with radius $r = R/2 + \Delta$. Then, depending on $\Delta$, there exist constants $\gamma, \eta$, such that 
\begin{align*}
     \mu(\mathcal{I}(r)) &\geq \left(1+ \Theta\left(e^{-\alpha R}\right)\right)\frac{\alpha e^{-\alpha r}}{(\alpha - 1/2)}\left(\frac{2}{\pi}e^{\frac{1}{2}(2\alpha - 1)(2r-R)} - \left(\frac{2}{\pi} - \frac{(\alpha - 1/2)}{\alpha}\right)\right) &&\text{and}\\ 
    \mu(\mathcal{I}(r))
     &\leq \left(1+ \Theta\left(e^{-\alpha R}\right)\right)\frac{\alpha e^{-\alpha r}}{\alpha - 1/2}\left(\gamma e^{\frac{1}{2}(2\alpha - 1)(2r-R)} - \eta\right),
\end{align*}

where
\begin{align*}
     \gamma, \eta =
\begin{cases}
1, \frac{1}{2\alpha} &\text{ for } \Delta \geq 0,\\
\frac{4}{3\sqrt{3}},\frac{1}{2\alpha} - \left(1-\frac{4}{3\sqrt{3}}\right)\left(\frac{4}{3}\right)^{(\alpha - 1/2)}  &\text{ for }  \Delta \geq \log(\sqrt{4/3}),\\
\frac{1}{\sqrt{2}},\frac{1}{2\alpha} - \left(1-\frac{4}{3\sqrt{3}}\right)\left(\frac{4}{3}\right)^{(\alpha - 1/2)}  -\left(\frac{4}{3\sqrt{3}} - \frac{1}{\sqrt{2}}\right)2^{(\alpha - 1/2)} &\text{ for } \Delta \geq \log(\sqrt{2}), \\
\frac{2}{3},\frac{1}{2\alpha} - \left(1-\frac{4}{3\sqrt{3}}\right)\left(\frac{4}{3}\right)^{(\alpha - 1/2)}  -\left(\frac{4}{3\sqrt{3}} - \frac{1}{\sqrt{2}}\right)2^{(\alpha - 1/2)} -\left(\frac{1}{\sqrt{2}} -\frac{2}{3}\right)2^{(2\alpha -1)} &\text{ for }  \Delta \geq \log(2). \\
\end{cases} 
 \end{align*}
\end{lemma}

\begin{lemma}\label{lem:radius-maximal-innerdegree}
Let $r^*$ be the radial coordinate of the point in \disk that maximises the measure of the inner-ball. Then $r^* = R/2 + \log\left(\frac{\alpha\eta}{\gamma(1-\alpha)}\right) /(2\alpha - 1)$.
\end{lemma}
\begin{proof}
Using the expression of $\mu(\mathcal{I}(r))$ in \cref{cor:inner-neighbourhood-equal},
\begin{align*}
   \derive{r} \mu(\mathcal{I}(r))
   &= \left(1+ \Theta\left(e^{-\alpha R}\right)\right)\left(\frac{\eta\alpha^2 e^{-\alpha r}}{\alpha - 1/2} - \frac{\gamma\alpha e^{-\alpha r}}{\alpha - 1/2}\left((1-\alpha)e^{\frac{1}{2}(2\alpha - 1)(2r-R)}\right)\right).
\end{align*}
Setting this equal to $0$ and solving for $r$ yields
$r^* = R/2 + \log(\frac{\alpha\eta}{\gamma(1-\alpha)}) / (2\alpha - 1)$.

\end{proof}

We use \cref{cor:inner-neighbourhood-equal} to upper and lower bound the degeneracy. The lower bound tells us that there exists a constant $\delta_{\alpha} > 0$ such that $\degen \geq (1+\delta_{\alpha})\coresize$ \whp The constant $\delta_{\alpha}$ is increasing with increasing $1/2 < \alpha<1$, see \cref{fig:plots}. 

\begin{restatable}[Bounds on the degeneracy]{theorem}{degeneracyUpper}\label{the:degeneracy-upper}
Let $G\sim\mathcal{G}(n, \alpha, C)$ be a threshold HRG. Then \whp~its degeneracy \degen satisfies 
\begin{align*}
\frac{(4-o(1))}{\pi}\left(\frac{2\left(1-\alpha\right)}{\frac{\pi}{2}-\alpha\left(\pi-2\right)}\right)^{\frac{1-\alpha}{2\alpha-1}}\coresize \leq \kappa(G) \leq ((4/3)^{\alpha} +o(1))\coresize.
\end{align*}
\end{restatable}

\begin{proof}[Proof sketch]
The upper bound is obtained by using $r^*$ as stated in \cref{lem:radius-maximal-innerdegree} for the upper bound of the inner-ball with $(n-1)\mu(\mathcal{I}(r^*))$, and using a Chernoff bound along with a union bound which gives an upper bound for $|\innerdeg(\ustar)|$ \whp \cref{obs:deg-upper-inner} then yields the upper bound for the degeneracy. For the lower bound, we first show that there exists a vertex with a radius $\Tilde{r}$ close in value to $r^*$ \whp Then $(1-o(1))(n-1)\mu(\mathcal{I}(\Tilde{r}))$ lower bounds $|\innerdeg(\ustar)|$ \whp This gives a lower bound for the degeneracy, due to \Cref{lem:deg-lower-inner}.
\end{proof}

Applying \cref{lem:lower-clique-number,the:degeneracy-upper}, the following is immediate.
\begin{corollary}[Bounds on the chromatic number]\label{cor:chromatic}
    Let $G\sim\mathcal{G}(n, \alpha, C)$ be a threshold HRG. Then \whp its chromatic number is $\coresize \leq \chrom \leq ((4/3)^{\alpha} +o(1))\coresize.$
\end{corollary}

Our structural results directly produce algorithmic applications. The small gap between degeneracy and core translates into an efficient approximation algorithm to colour a HRG.

\begin{restatable}[Approximation algorithm]{theorem}{approxAlgo}\label{the:approx-algo}
 Let $G\sim\mathcal{G}(n, \alpha, C)$ be a threshold HRG. Then an approximate vertex colouring of $G$ can be computed in time $\mathcal{O}(n)$ with approximation ratio $((4/3)^{\alpha}+o(1))$ \whp
\end{restatable}
\begin{proof}[Proof sketch.] Using a \emph{smallest-last} vertex ordering \cite{Matula1983} the number of colours required is upper-bounded by \degen. Computing the smallest-last vertex ordering, and then using the ordering to colour the graph, both takes linear as the giant component is sparse~\whp by \cite[Corollary 17]{hrg-spectral}. The approximation ratio is achieved by comparing the lower bound of the chromatic number in~\cref{cor:chromatic} to the upper bound of the degeneracy in \cref{the:degeneracy-upper}.
\end{proof}

\section{Clique Number of Hyperbolic Random Graphs}\label{sec:clique-number}
Recall that for any graph $G$, the clique number \clique, chromatic number \chrom, and degeneracy \degen are related via the inequalities $\clique\le\chrom\le\degen+1$. For this reason we are interested in the relationship between clique number and degeneracy for hyperbolic random graphs. In \cref{sec:clique-deg-gap} we show that the two differ and that for HRGs, the degeneracy is strictly larger than the clique number by a constant multiplicative constant. In \cref{sub:position-clique-hrg} we give new insights about where in the hyperbolic disk the largest clique is formed. We then conclude the section by providing a new upper bound for the clique number in \cref{sec:clique-upper-bound} that states a leading constant in front of the size of the core, and that is increasing in $\alpha$.

\subsection{The gap between Clique Number and Degeneracy}\label{sec:clique-deg-gap}
Because of the centralising effect of hyperbolic geometry, one might hope to show that the clique contained in the core of the disk is the largest, and that $\omega(G) = (1-o(1))\kappa(G)$. 
This would achieve two things on HRGs: first, it would imply a tight bound for the chromatic number $\chi(G)$, sandwiching it between clique number and degeneracy. Moreover, it would also imply a linear time $(1+o(1))$-approximation algorithm for the two NP-complete problems clique number and chromatic number using a smallest-last vertex ordering (see \cite{Matula1983}). 

In this section we disprove these claims. We show that there exists a constant gap between clique number and degeneracy; this is the content of \cref{the:clique-deg-gap}. Before embarking on the details of the proof, we first sketch its idea. For any clique $K$, its size is bounded by the \indeg $|\innerdeg(u)|$, where $u$ is the vertex with largest radius among the vertices of $K$. For $r(u) \leq R/2 + o(1)$ and $ r(u) \in R/2 + \omega(1)$, $|\innerdeg(u)|$
is already smaller by a multiplicative constant than the lower bound for the degeneracy given in \cref{the:degeneracy-upper}. What remains is to extend the result to $r(u)\in R/2 + \Theta(1)$, which requires more intricate arguments and is addressed in the following lemma.

\begin{lemma}\label{lem:case3}
   Let $G\sim\mathcal{G}(n, \alpha, C)$ be a threshold HRG, let $\Delta \in \Theta(1)$ and let $K$ be any clique where $u \in K$ is the vertex with largest radius $
   r_K = R/2 + \Delta$. Then \whp~there exists a constant $\varepsilon \in(0,1)$ such that $|K| \leq (1-\varepsilon)|\innerdeg(u)|$.  
\end{lemma}

\begin{proof}
  \begin{figure}[t]
    \centering \includegraphics[height=0.24\textheight]{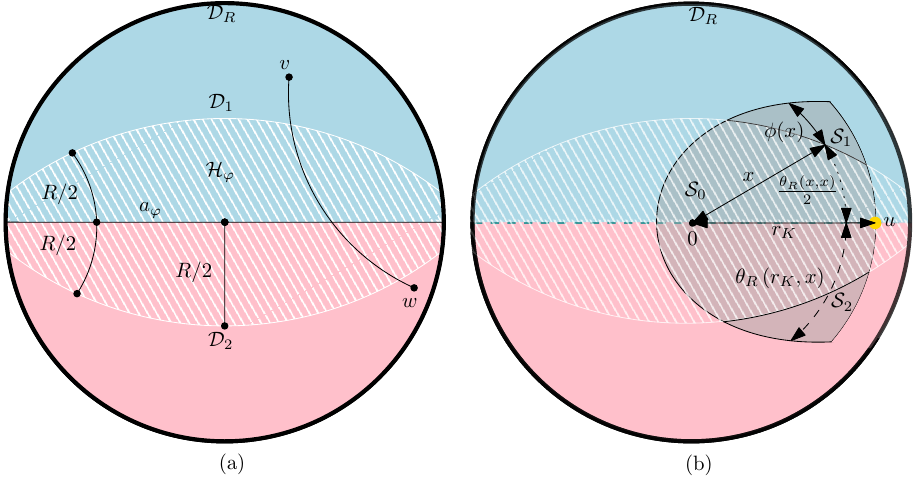}
    \caption{Illustration of a separator. (a) The line $a_{\varphi}$ partitions disk \disk into two halfdisks $\mathcal{D}_1$ (blue) and $\mathcal{D}_2$ (pink). The hypercycle $\mathcal{H}_{\varphi}$ (hatched area) is defined by the line $a_{\varphi}$. Points located in different halfdisks and outside the hatched area have distance at least $R$. (b) The separator (hatched area) separates the \inner (grey area) of a vertex $u$ of radius $r_K$ into three sub-areas $\mathcal{S}_0$, $\mathcal{S}_1$ and $\mathcal{S}_2$. Any vertex located in $\mathcal{S}_1$ has no edge to a vertex in $\mathcal{S}_2$.}  
    \label{fig:separator}
  \end{figure}
   
Let $G'$ be the \inner of $u$, i.e., the induced subgraph $G' = G[\Gamma(u)]$. Then $K \subseteq V(G')$, and thus $\omega(G') \ge |K|$. We show that $\omega(G') < (1-\varepsilon)|\innerdeg(u)|$ \whp

We accomplish this by proving that there exists a colouring for $G'$ with $(1-\varepsilon)|V(G')| = (1-\eps)|\innerdeg(u)|$ many colours. This implies an upper bound for the chromatic number $\chi(G')$, which also serves as an upper bound for $\omega(G')$. 
We do this by partitioning the \inner $\mathcal{I}_u$ into three disjoint sub-regions
$\mathcal{S}_0 \cup \mathcal{S}_1 \cup \mathcal{S}_2$ such that no vertex in $\mathcal{S}_1$ is adjacent to any vertex in $\mathcal{S}_2$. Thus, $\mathcal{S}_0$ separates $\mathcal{S}_1$ from $\mathcal{S}_2$ and we can colour the set of vertices $V \cap \mathcal{S}_1$ with the same colours as $V \cap \mathcal{S}_2$. Thus if $\min\{|V \cap \mathcal{S}_1| , |V \cap \mathcal{S}_2|\} \geq \varepsilon|\innerdeg(u)|$ \whp, our desired statement will be proven.

To find such a separator $\mathcal{S}_0$, we follow the lines drawn by Bläsius et.~al.~\cite{bfk-tw-2016} where \emph{hypercycles} for hyperbolic random graphs were introduced (see \cref{fig:separator}~a). A hypercycle $\hypercycle_\varphi$ (of radius $R/2$) is defined as follows: let $a_{\varphi}$ denote the line whose points have angle $\varphi$ and $\varphi + \pi$. Then $\hypercycle_\varphi := \{u \in \disk : \dist(u,a_\varphi) \leq R/2 \}$, i.e.~the set of points with distance at most $R/2$ to line $a_\varphi$. 
Consider the point $u = \left(r_K, \varphi\right)$ and let $\mathcal{S}_0 := \mathcal{I}_u \cap \hypercycle_\varphi$. 
To define $\mathcal{S}_1$ and $\mathcal{S}_2$ separate \disk into the two disjoint \emph{halfdisks} $\mathcal{D}_1 = \{x \in \disk : \varphi(x) - \varphi \leq \pi\}$ and $\mathcal{D}_2 = \{x \in \disk : \varphi(x) - \varphi \geq \pi\}$ (see \cref{fig:separator}~a). Then we define $\mathcal{S}_1 := (\mathcal{I}_u \cap \mathcal{D}_1) \setminus \hypercycle_\varphi$ and symmetrically $\mathcal{S}_2 := (\mathcal{I}_u \cap \mathcal{D}_2) \setminus \hypercycle_\varphi$. 

We observe that any point $w \in \mathcal{S}_1$ has distance at least $R$ to any point $v \in \mathcal{S}_2$. This can be shown for example by observing that the geodesic between $w$ and $v$ must pass through some point $x\in a_\varphi$, and so $\dist(w,v) = \dist(w,x) + \dist(x,v) \ge \dist(w,a_\varphi) + \dist(a_\varphi,v) > R$. This ensures our objective of separating the two regions via $\mathcal{S}_0$. 

Next, we show that there exists a constant $\varepsilon > 0$ small enough, such that $\mu(\mathcal{S}_1) = \mu(\mathcal{S}_2) \geq \varepsilon n^{-\alpha}$
(a sketch of the idea is given in \cref{fig:separator}~b). Setting $\phi(x) = \max(0, \theta_R\left(r_K,x\right) - \theta_R(x,x)/2)$ we derive by symmetry
$$
    \mu(\mathcal{S}_1) = \mu(\mathcal{S}_2) 
    = \int_{R/2}^{r}\int_{0}^{\phi(x)} \rho(x)\dd\phi \dd x = \int_{R/2}^{R/2 + \Delta} \phi(x) \rho(x) \dd x.
$$
Now we choose another constant $\Delta' < \Delta$ that fulfills $2\theta_R(R/2 + \Delta, R/2 + \Delta') - \theta_R(R/2 + \Delta', R/2 + \Delta') \eqqcolon c \in \Theta(1)$.
This is possible because $\Delta\in\Omega(1)$.
By our choice of $\Delta'$ we then obtain
\begin{align*}
    \mu(\mathcal{S}_1) &\geq \int_{R/2 + \Delta'}^{R/2 + \Delta} \phi(x)\rho(x) \dd x= \int_{R/2 + \Delta'}^{R/2 + \Delta} \left(\theta_R(R/2 + \Delta, x) - \frac{\theta_R(x,x)}{2}\right) \rho(x) \dd x\\
    &\geq c \int_{R/2 + \Delta'}^{R/2 + \Delta} \rho(x) \dd x 
    = \frac{c(\cosh(\alpha(R/2 + \Delta)) - \cosh(\alpha(R/2 + \Delta')))}{\cosh(\alpha R) - 1} \in \Omega(n^{-\alpha}),
\end{align*}
where the last line follows since $\theta_R(\cdot,\cdot)$ is monotonically decreasing and since $\rho(x) = \frac{\alpha\sinh(\alpha x)}{\cosh(\alpha R) - 1}$, $|\Delta - \Delta'| > 0$ and $R = 2\log(n) + C$. Therefore $\mu(\mathcal S_1) \in \Omega(\E{|\innerdeg(u)|}/n)$ \whp, since
$$
    \liminf_{n\to\infty} \frac {n\mu(\mathcal S_1)}{\E{|\innerdeg(u)|}} 
    = \liminf_{n\to\infty} \frac{\mu(\mathcal S_1)}{n^{-\alpha}}\frac{n^{1-\alpha}}{\E{|\innerdeg(u)|}}
    > 0,
$$
where $\E{|\innerdeg(u)|} \leq n\mu(\mathcal{B}_0(R/2 + \Delta)) \in \mathcal{O}(n^{1-\alpha})$ by \cref{eq:measure-of-origin-ball}, because $\Delta\in\mathcal O(1)$. Thus there exists some $\eps>0$ for which, for $n$ large enough, $$\E{\abs{V\cap\mathcal S_2}} = \E{\abs{V\cap\mathcal S_1}} = n\mu(\mathcal S_1) \ge (1-1/\log(n))^{-2}\eps\E{|\innerdeg(u)|};$$ since $\abs{V\cap\mathcal S_1} \le |\innerdeg(u)|$ a.s.~and is strictly smaller with positive probability, then $\eps< 1$.

Using a Chernoff bound 
for both $\abs{V\cap\mathcal S_1}$ and $\abs{V\cap\mathcal S_2}$, we obtain that neither random variable is smaller than $(1-1/\log(n))\E{\abs{V\cap\mathcal S_1}} \geq (1-1/\log(n))^{-1}\eps\E{|\innerdeg(u)|}$ \whp
On the other hand, another application of a Chernoff bound reveals $|\innerdeg(u)| \leq (1+ 1/\log(n))\E{|\innerdeg(u)|}$ \whp, since for $r(u) \in R/2 + \Theta(1)$ we have $\E{|\innerdeg(u)|} \in \Theta(n^{(1-\alpha)})$ using \cref{cor:inner-neighbourhood-equal}. A union bound 
then shows that \whp, $
    \min\{|V \cap \mathcal{S}_1| , |V \cap \mathcal{S}_2|\}
    \ge \frac {\eps\E{|\innerdeg(u)|}}{1 - 1 /\log(n)}
    \ge \eps|\innerdeg(u)|.
$
As argued above, a naïve colouring that colours the vertices of $\mathcal{S}_1$ and $\mathcal{S}_2$ with the same set of colours yields the upper bound.
\end{proof}

\begin{theorem}[Clique-degeneracy-gap]\label{the:clique-deg-gap}
Let $G\sim\mathcal{G}(n, \alpha, C)$ be a threshold HRG. Then \whp~there exists a constant $\varepsilon \in(0,1)$ such that $\kappa(G)/\omega(G) > 1 + \varepsilon$.   
\end{theorem}
\begin{proof}
Let $K$ be the largest clique of $G$, and let $u$ be the vertex of $K$ with maximal radius $r_K\coloneqq r(u)$. We assume the following cases for $r_K$ which cover all possibilities:

\noindent\textbf{Case 1}~[$r_K \in R/2 + \omega(1)$]: 
Observe that $K \subseteq V \cap \mathcal{I}\left(r_K\right)$. Hence, $|K| \leq |V \cap \mathcal{I}\left(r_K\right)|$ and
\begin{align*}
        \mu\left(\mathcal{I}\left(r_K\right)\right) \leq \left(1+ \Theta\left(e^{-\alpha R}\right)\right)\frac{\alpha e^{-\alpha r_K}}{\alpha - 1/2}\left(\gamma e^{\frac{1}{2}(2\alpha - 1)(2r_K-R)} - \eta\right) \in \Theta(1)n^{-\alpha}e^{-(1-\alpha)\omega(1)},
\end{align*}
by~\cref{cor:inner-neighbourhood-equal} since $\gamma$ and $\eta$ are both constants. Taking the expectation and using a Chernoff bound 
we have $|K| \in o(n^{1-\alpha})$~\whp Recall that $\kappa(G) > (1+\delta)e^{-\alpha C/2}n^{1-\alpha}$ \whp~by \cref{the:degeneracy-upper}. It follows $\degen/|K| \in \omega(1)$~\whp 

\noindent\textbf{Case 2}~[$r_K \leq R/2 + o(1)$]: 
Observe that the size of any clique $K \subseteq V \cap \mathcal{B}_0(r)$ is upper bounded by $X = |V \cap \mathcal{B}_0(r)|$.
By \cref{eq:measure-of-origin-ball} we have $\mu(\mathcal{B}_0(r_K)) \leq (1+o(1))n^{-\alpha}e^{-\alpha C/2}$. Hence, we get $\mathbb{E}[X] \leq (1+o(1))n^{1-\alpha}e^{-\alpha C/2}$. Since $X$ is a binomial random variable we can apply a Chernoff bound, which yields 
$|K| \leq X \le (1+o(1))n^{1-\alpha}e^{-\alpha C/2}$ \whp
Taking $\eps = \delta/(1+\delta)$ with $\delta$ as in \cref{the:degeneracy-upper}, we have
$$\frac {\abs{K}}{1-\eps} 
    \le (1+o(1))(1+\delta)n^{1-\alpha}e^{-\alpha C/2}
    < (1+o(1))\kappa(G)\text{ \whp}$$

\noindent\textbf{Case 3}~[$r_K \in R/2 + \Theta(1)$]: Let $\xi \in o(1)$. Using \cref{lem:case3,lem:deg-lower-inner}, we get \whp~that $$\clique = |K| \leq \frac{1-\varepsilon}{1-\xi}|\innerdeg(u)| \leq \frac{1-\varepsilon}{1-\xi}|\innerdeg(\ustar)| \leq (1-\varepsilon)\degen,$$
for an adequate choice of $\xi$.
\end{proof}

\begin{corollary}\label{cor:clique-upper-bound}
Let $G\sim\mathcal G(n,\alpha, C)$ be a threshold HRG.
Then there exists a constant $\eps\in(0,1)$ such that \whp,
$
   \clique \leq (4(1-\eps)/3)^{\alpha} \coresize. 
$
\end{corollary}
\begin{proof}
This follows from \cref{the:degeneracy-upper,the:clique-deg-gap}.
\end{proof}

\subsection{Cliques larger than the Core}\label{sub:position-clique-hrg}

In this section, we show that there exist a super-constant number of unique cliques that contain the core, but are strictly larger than it. The overall argument goes as follows. We consider a set of vertices with radial coordinates slightly outside the core. We show that any vertex of this set has a constant probability to be adjacent to all vertices belonging to the core, and thus induces a clique that is larger than the core itself.  To this end, the following lemma concerned about points close to the core proves useful.
\begin{restatable}{lemma}{forbiddenAngle}\label{lem:theta-lower-xi-gap}
Let $k \in \mathbb{N} \setminus \{0\}$ and $\xi_k = \log(1 + \frac{\log^k(n)}{n^{1-\alpha}}) \in o(1)$. Consider two points with radial coordinates $r = R/2 + \xi_k$ and $x = R/2$. Then $\theta_R(r,x) \geq \pi - 2\sqrt{\log^k(n) n^{\alpha-1}}$.     
\end{restatable}
Thus, when $r$ approaches $R/2$ from above, the angle $\theta_R(r, R/2)$ approaches $\pi$ from below.

\begin{proposition}\label{prop:larger-cliques}
Let $G\sim\mathcal{G}(n, \alpha, C)$ be a threshold HRG. 
Then \whp there exist $\omega(\log(n))$ cliques that are larger than $\coresize$. 
\end{proposition}
\begin{proof}
  \begin{figure}[t]
    \centering \includegraphics[height=0.24\textheight]{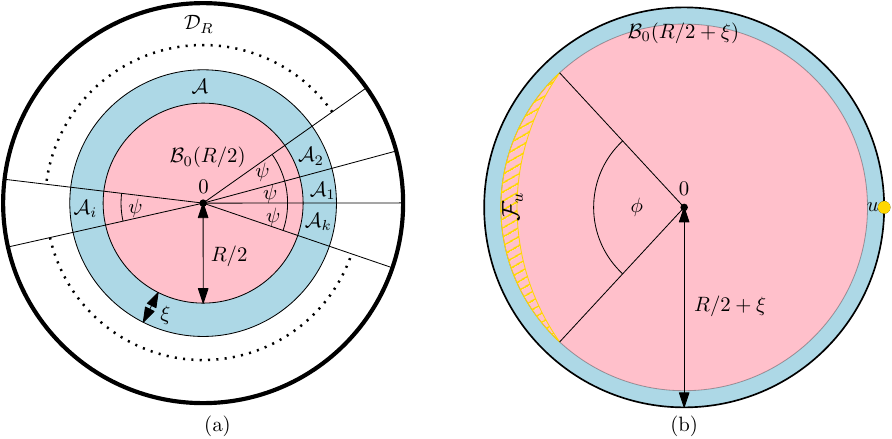}
    \caption{Sketch of the proof idea for \cref{prop:larger-cliques}. (a) Illustration of the set of points $\mathcal{A}$ (blue) of width $\xi$ slightly outside of the core (pink). The intersection with a sector of angle $\psi$ and the band $\mathcal{A}$ forms a box $\mathcal{A}_i$ that contains a vertex \whp The number of non-intersecting boxes is $k = 2\pi/\psi \in \omega(\log(n))$. 
    (b) A vertex $u$ located on the boundary of the area $\mathcal{A}$. The hatched area $\mathcal{F}_u$ with angle at most $\phi$ is the corresponding forbidden area of $u$. Any point in $\forbidden$ has distance at least $R$ to $u$. An adequate choice for the width $\xi$ yields that this area is empty with constant probability.}
    \label{fig:forbidden-area}
  \end{figure}
For the proof we  consider the Poissonized version of the HRG model (see e.g. \cite{Kiwi2019, Kiwi2024}). The upshot of this model is that it allows us to analyse disjoint areas in the hyperbolic disk independently.
Since the final result holds \whp this directly carries over to the uniform model \whp (\cite[Lemma 3.9]{Katzmann2023}). We start by defining an area $\mathcal{A}$ close to the core. Let $\xi = \log\left(1 + \frac{\log^3(n)}{n^{1-\alpha}}\right) \in o(1)$ and consider a band of points $\mathcal{A} := \mathcal{B}_0(R/2 + \xi)\setminus\mathcal{B}_0(R/2)$. We see via $R = 2\log(n) + C$ that
\begin{align}\label{eq:bound-band}
\begin{split}
    \mu(\mathcal{A}) 
    &= \int_{R/2}^{R/2 + \xi} \frac{\sinh(\alpha x)}{\cosh(\alpha R) - 1} 
    = \frac{\cosh(\alpha(R/2 + \xi)) - \cosh(\alpha R/2)}{\cosh(\alpha R) - 1}\\
    &= (1+\Theta(e^{-\alpha R}))(e^{-\alpha(R/2 - \xi)} - e^{-\alpha R/2})\\
    &= (1+\Theta(e^{-\alpha R}))(n^{-\alpha}e^{-\alpha C/2}{\left(1+\log^3(n)n^{\alpha-1}\right)} - n^{-\alpha}e^{-\alpha C/2})\\
    &= (1+\Theta(e^{-\alpha R}))\big(e^{-\alpha C/2}\log^3(n)n^{-1}\big) \in \Theta\left(\log^3\left(n\right)/n\right).
\end{split}
\end{align}
Let $A = V \cap \mathcal{A}$. Then $\mathbb{E}[|A|]\in \Theta\left(\log^3(n)\right)$. We further partition the band $\mathcal{A}$ into $k = \lceil\log^{3/2}(n)\rceil$ sectors $\mathcal{A}_1,\ldots,\mathcal{A}_k$, each of equal size (see \cref{fig:forbidden-area}~a), and 
$A_i = V \cap \mathcal{A}_i$. Since $\mathbb{E}[|A|]\in \Theta\left(\log^3(n)\right)$ and $k = \log^{3/2}(n)$, we have for any $i \in [k]$ that $\Exp[|A_i|] = \frac{\mathbb{E}[|A|]}{k} \in \omega(\log(n))$.
Since we are in the Poissonized model we get 
$
\Prob{|A_i| =  0} = e^{-\E{|A_i|}} \in n^{-\omega(1)}.
$
Subsequently, a union bound yields that there is no sector $\mathcal{A}_i$ that is empty \whp

In the second step, we show that for any vertex $u \in A$, the probability that there exists a vertex in its \emph{forbidden area} $\forbidden := \{x \in \B_0(R/2) : \dist(u,x)> R\}$ (see \cref{fig:forbidden-area}~b) is strictly less than $1$. Since $r(u) \leq R/2 + \xi$, we have that $u \in A$ has distance at most $R$ (and thus, an edge) to any vertex in the area $\mathcal{B}_0(R/2 - \xi)$. Hence, we have $\forbidden \subset \mathcal{B}_0(R/2) \setminus \mathcal{B}_0(R/2 - \xi)$. Now, for $R/2 \geq x \geq R/2 - \xi$ we seek to find the angle size $\phi = 2(\pi - \theta_R(r(u),x))$ in order to upper bound $\mu(\forbidden)$. Since $\theta_R(r,x)$ is monotonically decreasing in $x$ we have $\phi \leq 2(\pi - \theta_R(R/2 + \xi, R/2))$ and we obtain
\begin{align*}
    \mu(\forbidden) 
    = \int_{R/2 - \xi}^{R/2}\int_0^{\phi} \rho(x)\dd\phi \dd x
    &\leq 2(\pi - \theta_R(R/2 + \xi, R/2)) \int_{R/2}^{R/2 + \xi}\rho(x) \dd x\\ 
    &\leq 2(\pi - \theta_R(R/2 + \xi, R/2)) \mu(\mathcal{A}),
\end{align*}
where we used $\int_{R/2 - \xi}^{R/2}\rho(x) \dd x \leq \int_{R/2}^{R/2 +\xi}\rho(x) \dd x = \mu(\mathcal{A})$.
By our choice of $\xi = \log(1 + \log^3(n)n^{\alpha-1}) \in o(1)$, we can apply \cref{lem:theta-lower-xi-gap} and obtain $\theta_R(R/2 + \xi, R/2) \geq \pi - 2\sqrt{\frac{\log^3(n)}{n^{1-\alpha}}}$. In \cref{eq:bound-band} we established that $\mu(\mathcal{A}) \in \Theta\left(\log^3\left(n\right)/n\right)$. Combining the two leads to $\mu(\forbidden) \in \mathcal{O}\Big(\sqrt{\log^9(n)n^{\alpha-3}}\Big)$.

Notice that for $\alpha \in (1/2, 1)$, the measure of the forbidden area of a vertex $u \in A$ is $\mu(\forbidden) \in o(1/n)$. Hence, writing $F = V \cap \forbidden$,
we get $\mathbb{E}[|F|] \in o(1)$, i.e., the expected number of vertices in $\forbidden$ is vanishing. Applying Markov's inequality then gives us  $\Prob{|F| \geq 1} \in o(1)$. Thus $p \coloneqq \P(\abs{F} < 1) \in \Omega(1)$, so the forbidden area is empty with constant probability.

We now establish our third and final desired property. Namely, we construct a subset $U \subset A$, consisting of vertices whose forbidden areas are empty and disjoint and for which $\E{\abs{U}} \in \omega(\log(n))$.

Recall that we partitioned $\mathcal{A}$ into $k = \log^{3/2}(n) \in \omega(\log(n))$ sectors. Let $X_i$ be the indicator that there exists a vertex $u\in\mathcal A_i$ whose forbidden area $\forbidden$ is empty.
Then $X = \sum_{i=1}^{\floor{k/2}} X_{2i}$ is a loose lower bound on the 
number of sectors with this property. By linearity of expectation, $p$ constant and $k = \log^{3/2}(n)$ we then obtain 
$$
    \mathbb{E}[X] \geq \mathbb{E}\Big[\sum_{i=1}^{\floor{k/2}}X_{2i}\Big] \geq p\cdot\Theta(1)\log^{3/2}(n) \in \omega(\log(n)).
$$
We proceed by showing independence among the random variables $X_i$ and $X_j$ for $\abs{j-i}>1$. To this end, observe that the angle $\psi$ (see \cref{fig:forbidden-area}~b) spanned by any sector $\mathcal{A}_i$ is $\psi = 2\pi/k \geq \left\lfloor2\pi\log^{-3/2}(n)\right\rfloor$. In contrast, we recall $\phi \leq 2(\pi - \theta_R(R/2 + \xi, R/2)) \leq 4\sqrt{\log^3(n)n^{\alpha-1}}$.
Since $\sqrt{\log^3(n)n^{\alpha-1}} \in o(\log^{-3/2}(n))$ we conclude $\phi \in o(\psi)$. 
This implies that the forbidden areas $\mathcal F_u$ and $\mathcal F_v$ of any $u\in \mathcal A_i$, $v\in\mathcal A_j$ are disjoint. Thus $X_i$ and $X_j$ are independent.

To wrap things up, recall that in the first step we established that each sector $\mathcal{A}_i$ contains a vertex \whp Moreover, since the $X_i$ are independent, we have by a Chernoff bound 
that $X \in \omega(\log(n))$ \whp Though these two events are not independent, we can apply the union bound to their complements to obtain that \whp $\omega(\log(n))$ vertices outside of $\mathcal{B}_0(R/2)$ are adjacent to all vertices in $\mathcal{B}_0(R/2)$, which finishes the proof.   
\end{proof}

\subsection{Upper Bound on the Clique Number}\label{sec:clique-upper-bound}

Recall that two vertices $u,v$ are adjacent if and only if $\dist(u,v) \leq R$ and that we call the clique formed in $\mathcal{B}_0(R/2)$ the core whose size is $\coresize = (1-o(1))e^{-\alpha C/2}n^{1-\alpha}$ \whp The core size is a lower bound for the clique number $\clique$. We have established that the largest clique is smaller than the degeneracy \whp (\cref{the:clique-deg-gap}), and in this section we further investigate an upper bound for \clique. We note that the upper bound we derive in this section implies \cref{the:clique-deg-gap} for $\alpha$ large enough (see~\cref{fig:plots}). However, for smaller $\alpha$, the upper bound for $\clique$ is larger than the lower bound (\cref{the:degeneracy-upper}) for \degen in the HRG model, and thus does not directly imply \cref{the:clique-deg-gap} for these values for $\alpha$.

Before going into details, we lay out our proof strategy. We aim to bound the region where a clique can be located. Since vertices are adjacent if and only if their (hyperbolic) distance is at most $R$, this can be done by characterising a shape that covers any hyperbolic region of diameter $R$. 
A classic result by Jung \cite{Jung1901} answers the question of how large the radius of a ball in Euclidean space needs to be at most, so that its interior can contain an entire set of points of fixed diameter.
The hyperbolic version of this result was discovered by Dekster \cite{Dekster1995, Dekster1997} nearly a century later. He extended Jung's result to (among other geometries) hyperbolic space and we apply it as follows: we identify $\mathcal{O}(n^3)$ many balls where one of these balls contains the clique of largest size \clique. This clique (and all the other identified cliques) needs to be located in a ball $B_x(r)$ with radius $r$ large enough. We use the hyperbolic variant of Jung's theorem to upper bound $r$ which, in turn, allows us to upper bound the area of this ball. This yields an upper bound for the amount of vertices one such ball $B_x(r)$ could contain \whp, leading to an upper bound for \clique. Since we only need to consider at most $\mathcal{O}(n^3)$ balls, a union bound is sufficient to derive the same bound for the worst case. We work with the following version of Jung's theorem for hyperbolic geometry.

\begin{theorem}\cite[Theorem 2]{Dekster1995}\label{the:jung-hyperbolic}
    Let $\mathcal{K} \subset \mathbb{H}^d$ be compact and suppose that for any $y,z \in \mathcal{K}$, $\dist(y,z) \leq D$. Then there exists $x \in \mathbb{H}^d$ such that $\mathcal{K}\subseteq\mathcal{B}_x(r)$ for $r$ satisfying 
    $$
        D\geq 2\sinh^{-1}\left(\sqrt{\frac{d+1}{2d}}\sinh(r)\right).
    $$
\end{theorem}

In the hyperbolic plane $\mathbb{H}^2$, this simplifies to the following.

\begin{restatable}{corollary}{jungshyper}\label{cor:working-horse}
    Let $\mathcal{K} \subset \mathbb{H}^2$ be compact and suppose that for any $y,z \in \mathcal{K}$, $\dist(y,z) \leq R$. Then there exists $x \in \mathbb{H}^2$ such that $\mathcal{K}\subseteq\mathcal{B}_x(r)$ for $r$ satisfying $r\leq R/2 + \log(2/\sqrt{3})$.
\end{restatable}

\begin{proof}
Using that $d = 2$, we directly get from Theorem~\ref{the:jung-hyperbolic} for diameter $R$ that
$
    \frac{R}{2}\geq\sinh^{-1}(\sqrt{3/4}\sinh(r)).
$
Rearranging and using for $x \in \mathbb{R}$ that $\sinh(x) = \frac{1}{2}e^x(1-e^{-2x})$ yields
$$
    R/2 - r \geq \log(\sqrt{3/4}) + \log\left(\frac{\left(1 - e^{-2r}\right)}{\left(1 - e^{-R}\right)}\right).
$$
Solving for $r$ and using that $r \geq R/2 \geq \log(n) + C/2$ in conjunction with recalling that $C \in \Theta(1)$, it follows that
$
    r \leq R/2 + \log(2/\sqrt{3}).
$
\end{proof}

Our next observation follows from the definition of $\mu$, and formalises the intuition that $\mu$ puts more mass at the centre of the disk.
\begin{observation}\label{obs:area-of-disk}
Let $0<r\leq R$ and $u,v \in \disk$ with $r(u)\geq r(v)$. Then $
    \mu(\mathcal{B}_u(r)) \leq \mu(\mathcal{B}_v(r))$.
\end{observation}
   
Recall that \coresize denotes the core size $|V \cap \mathcal{B}_0(R/2)|$ which is a lower bound for the clique number $\clique$ (see \cref{lem:lower-clique-number}), and that $\coresize = (1-o(1))e^{-\alpha C/2}n^{1-\alpha}$ \whp We state our upper bound relative to this lower bound.

\begin{theorem}[Clique upper bound]\label{the:clique-upper-bound}
    Let $G\sim\mathcal{G}(n,\alpha,C)$ be a threshold HRG with $\alpha \in (1/2,1)$. Then \whp 
    $$
        \clique 
        \leq \left((4/3)^{\alpha/2}+o(1)\right)\coresize.
    $$
\end{theorem}

\begin{proof}
Consider any triplet of vertices $u,v,w \in V$ with pairwise distance at most $R$, so that they are pairwise adjacent. 
Since $u,v,w$ are a.s.~in general position, there is a unique ball $\mathcal{B}_x(r)$ such that $u,v,w$ lie on the boundary $\partial\mathcal{B}_x(r)$.
By \cref{cor:working-horse}, $r \leq R/2 + \log(\sqrt{4/3})$ since $\max{(\dist(u,v), \dist(v,w), \dist(u,w))} \leq R$. Over all possible triplets $u,v,w \in V$ this gives us a set $B$ of at most $n \choose 3$ closed balls. 
Any clique must be contained in one of these balls, and therefore so is the largest clique. Thus upper bounding the number of vertices for each individual ball $\mathcal{B} \in B$ yields an upper bound on the size of the largest clique. 

We now upper bound the expected number of vertices in one ball $B$.  To this end we fix a ball $\B = \mathcal{B}_x(r) \in B$ and let $Y_\B$ be the random variable counting the number of vertices in $\B$. The balls in $B$ are identically (though clearly not independently) distributed.
Since vertices are thrown independently according to $\mu$, we have that $Y_\B - 3\sim\mathrm{Bin}(n-3, \mu(\B))$, and so
\begin{align*}
    \E{Y_\B} 
    = 3 + (n-3)\mu(\B)
    &\leq 3 + (n-3)\mu(\mathcal{B}_0(r)) && \text{(\cref{obs:area-of-disk})}\\
    &\leq 3 + (n-3)\mu(\mathcal{B}_0(R/2 + \log(2/\sqrt{3}))) && \text{(\cref{cor:working-horse})}\\
    &\leq \big(\big(2/\sqrt{3}\big)^{\alpha}+o(1)\big)e^{-\alpha C/2}n^{1-\alpha}. && \text{(\cref{eq:measure-of-origin-ball})}
\end{align*}
Thus $((4/3)^{\alpha/2}+o(1))\coresize \ge (1+1/\log(n))\E{Y_\B}$ \whp
This is relevant to the bound in the theorem statement because it implies that
\begin{align*}
    \P\Big(\omega(G)> ((4/3)^{\alpha/2}+o(1))\coresize\Big)
    &\le \P\Big(\max_{\B\in B} Y_\B > ((4/3)^{\alpha/2}+o(1))\coresize\Big)\\
    &\le \P\Big(\max_{\B\in B} Y_\B > (1+1/\log(n))\mathbb{E}[Y_\B] \Big) + n^{-c}
\end{align*}
for arbitrary $c$.
Thus to finish the proof we need to show concentration,
which via a union bound over all triplets $u,v,w$ will yield the result.
To show concentration we apply a Chernoff bound \cref{lem:Chernoff}. Using $\varepsilon = (1/\log(n))$ we obtain 
\begin{align*}
    \Prob{Y_\B > (1+1/\log(n))\mathbb{E}[Y_\B]} &\leq e^{-\mathbb{E}[Y_\B]/(3(1/\log(n))^2)}
    \leq e^{-\Theta(1)(n^{1-\alpha})/\log^2(n)}
    \leq n^{-c}
\end{align*}
for any choice of $c$, since for $\alpha < 1$, $\Theta(1)(n^{1-\alpha})\in n^{\Theta(1)}$ and $\liminf_{n\to\infty} \frac{n^{\Theta(1)}}{\log^2(n)} \in \omega(\log(n)).$ Finally, to show that this holds \whp for all balls in $B$, we use that $|B| \leq 
{n\choose 3} < n^3$, so that
\begin{align*}
    \P\Big(\max_{\B\in B} Y_\B \geq (1+1/\log(n))\mathbb{E}[Y_\B]\Big) &\leq \sum_{\B\in B} \mathbb{P}\left( Y_\B \geq (1+1/\log(n))\mathbb{E}[Y_\B]\right)\\ 
    &\leq n^3\mathbb{P}\left( Y_\B \geq (1+1/\log(n))\mathbb{E}[Y_\B]\right)
    \leq n^{-c+3}.\qed
\end{align*}
\let\qed\relax
\end{proof}
A further refinement of the ``clique covering'' argument of \cref{the:clique-upper-bound} should be possible. Any clique has by definition a diameter of at most $R$, and so the shape in $\mathbb H^2$ of diameter $R$ with maximal area would provide an improved upper bound via a similar covering argument. It is not clear what a tight bound would be, and $\clique \leq (1+o(1))\coresize$ may be possible.

\section{Geometric Inhomogeneous Random Graphs}\label{sec:girg}
\emph{Geometric Inhomogeneous Random Graphs} or \emph{GIRGs} were introduced in \cite{BRINGMANN201935} as an alternative model to HRGs that capture many of the same properties, in particular the power-law degree distribution. In their most general form, GIRGs strictly generalise HRGs, but they are more often studied in a slightly restricted form; comparisons are made in \cite{blaesius-esa-2019,KOMJATHY20201309}. In this restricted form, called the \emph{standard} GIRG model by \cite{Schiller2024}, any HRG $G$ can be coupled with two GIRGs $H_1$ and $H_2$ such that $H_1\subseteq G \subseteq H_2$, where $\subseteq$ denotes graph inclusion.

Because of this relationship, GIRGs are used as proxies for HRGs in some theoretical and experimental works. This is partly done because GIRGs are (by design) far more tractable than HRGs. 
It is therefore valuable to understand differences between the two models.
In \cite{blaesius-esa-2019} experimental evidence was given to suggest that the ``sandwiching'' of an HRG by two standard GIRGs is not tight. In \cref{cor:hrg-girg-gap} we provide a theoretical result demonstrating a difference between the two models.

\begin{definition}[Standard GIRG model]\label{def:GIRG}

Let $\beta \in (2,3)$, $\lambda \in \Theta(1)$, and $n \in \mathbb{N}$. A \emph{geometric inhomogeneous random graph} $G \sim \GIRG$ is a random graph with vertex set $V = \{v_1,\cdots, v_n\}$ satisfying the following properties.
\begin{enumerate}
    \item Every $u \in V$ is equipped with a random tuple $(w_u, x_u)$, where weight $w_u \in [1, \infty)$ has density $f(y) = (\beta - 1)y^{-\beta}$
    and coordinate $x_u$ is drawn uniformly at random from $[0,1]$;
    \item Any pair of vertices $u,v\in V$ are connected if and only if $\min\{|x_u - x_v|, 1-|x_u - x_v|\} \leq t(u,v)$, where $t(u,v) = \frac{1}{2}\left(\frac{\lambda w_u w_v}{n}\right)$.
\end{enumerate}
\end{definition}

One way of thinking of a GIRG is that vertices are being thrown uniformly at random onto the 1-dimensional torus $\mathbb T^1$, and connected according to whether their distance is below their threshold $t(u,v)$. The weights are drawn according to a Pareto distribution. Analogously to HRGs, for a vertex $u$ of a GIRG we define the inner-degree of $u$ to be $|\innerdeg(u)| = \abs{\{v \in V | \text{$u$ and $v$ are connected and } w_v \geq w_u\}}.$ The proofs of \cref{obs:deg-upper-inner,lem:deg-lower-inner} can be adapted to the GIRG model to characterise the degeneracy via the largest inner-degree.

\begin{restatable}{corollary}{girgInner}\label{cor:inner-neighbour-upper-lower}
    Let $G \sim \mathcal{G}(n, 2\alpha + 1, \lambda)$ be a standard GIRG. Consider the vertex $u^*$ with the largest inner-degree in $G$. Then \whp $\kappa(G) = (1- o(1))|\innerdeg(u^*)|.$
\end{restatable}

\cref{cor:inner-neighbour-upper-lower} allows us to state a tight bound for the degeneracy in comparison to the core of the GIRG, which is defined to contain all vertices of weight $\hat{w} \ge \sqrt{n/\lambda}$, and has size $\coresize = (1\pm o(1))\lambda^{-\alpha} n^{1-\alpha}$ \whp This is analagous to the core of an HRG, which is the clique formed by vertices of radius at most $R/2$, regardless of their angular coordinates.
\begin{restatable}{theorem}{girgDegeneracy}\label{the:girg-degeneracy}
Let $G \sim \mathcal{G}(n, 2\alpha + 1, \lambda)$ be a threshold GIRG. The degeneracy is \whp $$\kappa(G) = (2 \pm o(1)) (2(1-\alpha))^{({1 - \alpha})/{(2\alpha - 1)}}\coresize.$$
\end{restatable}

\begin{proof}
    We bound the maximal inner-degree $|\innerdeg(\ustar)|$; the statement then follows from \cref{cor:inner-neighbour-upper-lower}.
    Notice that, independent of the geometric distance, a vertex with weight $w$ is adjacent to any vertex with weight $w'$ if $w' \geq \frac{n}{w\lambda}$ since 
$
   t(w, w') = \frac{\lambda w w'}{2n} \geq \frac{\lambda w \frac{n}{w\lambda}}{2n} = \frac{1}{2}$
which is the maximal distance between two points in the unit torus. Thus, using $\beta = 2\alpha + 1$, the probability that a vertex $v$ is in the \inner of a vertex $u$ with weight $w$ is 

\begin{align}\label{eq:prob-inner-girg}
    \Prob{v \in \Gamma(u)} &= \Prob{\{\{u, v\}\in E\} \cap \{W_v \geq w\}} \notag\\
    &= \Prob{W_v \geq \frac{n}{w\lambda}} + 2\int_{w}^{\frac{n}{w\lambda}}t(y,w)2\alpha y^{-(2\alpha + 1)} dy && (\text{$t(w,n/w\lambda) = 1/2$})\notag\\
    &= \left(\frac{n}{w\lambda}\right)^{-2\alpha} + \frac{2\alpha \lambda w}{n}\int_{w}^{\frac{n}{w\lambda}}y^{-2\alpha} dy && (\text{Pareto and threshold})\notag\\
    &= \left(\frac{n}{w\lambda}\right)^{-2\alpha} + \frac{2\alpha\lambda w}{n(2\alpha - 1)}\left(w^{1-2\alpha} - \left(\frac{n}{w\lambda}\right)^{1-2\alpha} \right) && \left(\int y^{-2\alpha}dy = \left[\frac{y^{1-2\alpha}}{1-2\alpha}\right]\right)\notag\\
    &= \left(\frac{n}{w\lambda}\right)^{-2\alpha} + \frac{\alpha}{\alpha - 1/2}\left(\frac{\lambda w^{2(1-\alpha)}}{n} - \left(\frac{n}{w\lambda}\right)^{-2\alpha} \right). 
\end{align}

Next, we calculate the value $w^*$, which maximises the expected inner-degree. To this end, we consider the probability measure of the inner-neighbourhood, take its derivative with respect to $w$ and set it equal to $0$. Differentiating yields
\begin{align*}
    \derive{w}  \Prob{v \in \Gamma(u)} 
    &= \frac{2(nw)^{-(2\alpha + 1)}(1 - \alpha)\alpha\left(n^{2\alpha} w^2\lambda - nw^{4\alpha}\lambda^{2\alpha}\right)}{2\alpha - 1}, 
\end{align*}
and solving for the maximum reveals
$
    w^* = (2(1-\alpha))^{\frac{1}{4\alpha - 2}}\sqrt{\frac{n}{\lambda}} \in \Theta(\sqrt{n}).
$

We plug in $w^*$ for the weight of $u$ denoted by $u^*$ into $\Prob{v \cap \Gamma(\ustar)}$ and get by \eqref{eq:prob-inner-girg} that
$$
    \Prob{v \in \Gamma(\ustar)} = 2(2(1-\alpha))^{({1 - \alpha})/{(2\alpha - 1)}}(n\lambda)^{-\alpha}. 
$$
Recalling that $\coresize = (1\pm o(1))\lambda^{-\alpha} n^{1-\alpha}$ \whp, the upper bound now follows from the expectation of $|\innerdeg(u^*)|$ and applying a Chernoff bound in conjunction with a union bound. The lower bound is established by showing that there exists a vertex within the range of weights $\Tilde{w} = [w^*(1 + n^{\alpha-1}\log^2(n))^{-1/(2\alpha)}, w^*]$ \whp and lower bound the inner-degree of such vertex. Using the Pareto distribution and $w^* \in \Theta(\sqrt{n})$, we calculate the probability for a vertex $u$ to belong to the range of weights $\Tilde{w}$. We obtain
\begin{align*}
    \Prob{W_u \in \Tilde{w}} &= \Prob{w^*(1 + n^{\alpha-1}\log^2(n))^{-1/(2\alpha)} \leq W_u \leq w^*} &&(\text{Range of } \Tilde{w})\\
    &= \Prob{W_u \leq w^*} - \Prob{W_u \leq w^*(1 + n^{\alpha-1}\log^2(n))^{-1/(2\alpha)}}\\ 
    &= 1 - (w^*)^{-2\alpha} - (1 - (w^*(1 + n^{\alpha-1}\log^2(n))^{-1/(2\alpha)})^{-2\alpha}) &&(\text{Pareto})\\
    &= (w^*(1 + n^{\alpha-1}\log^2(n))^{-1/(2\alpha)})^{-2\alpha} - (w^*)^{-2\alpha}\\
    &= (w^*)^{-2\alpha}n^{\alpha-1}\log^2(n)\\
    &= \Theta(1)\frac{\log^2(n)}{n}. && (w^* \in \Theta(\sqrt{n}))
\end{align*}

By this we have $\E{\abs{V \cap \Tilde{w}}} \in \omega(\log(n))$. Using a Chernoff bound there exists a vertex within the desired weight range $\Tilde{w}$ \whp To conclude the proof we lower bound the inner-degree of a vertex $\Tilde{u}$ included in the weight range $\Tilde{w}$. Note that $\Tilde{w} = (1-o(1))w^* = (2-o(1)(1-\alpha))^{\frac{1}{4\alpha - 2}}\sqrt{\frac{n}{\lambda}}$. We then have via \Cref{eq:prob-inner-girg}
$$
    \E{|\innerdeg(\Tilde{u})|} = (n-1)\Prob{v \cap \Gamma(\Tilde{u})} 
    \geq (2-o(1))(2(1-\alpha))^{({1 - \alpha})/{(2\alpha - 1)}}\lambda^{-\alpha}n^{1-\alpha}.
$$
A final application of a Chernoff bound then ensures the concentration to finish the proof.
\end{proof}

Comparing the lower bound of the degeneracy for GIRGs given in \cref{the:girg-degeneracy} to the upper bound of a HRG we obtained in \cref{the:degeneracy-upper} we draw the conclusion that the degeneracy-to-core ratio between the two models is fundamentally different.

\begin{restatable}[GIRG-HRG degeneracy difference]{corollary}{girgGap}\label{cor:hrg-girg-gap}
    Fix an $\alpha \in (1/2, 1)$. Let $G \sim \mathcal{G}(n, 2\alpha + 1, \lambda)$ be a standard GIRG and $H \sim \mathcal{G}(n, \alpha, C)$ be a threshold HRG. Then  \whp\ 
    $$ 
 \left|\frac{\degen}{\coresize} -  \frac{\kappa(H)}{\sigma(H)}\right| 
    \in \Theta(1). 
    $$
    
\end{restatable}

\section{Conclusion}
We have shown that the clique number, degeneracy, and chromatic number of HRGs are asymptotically (with small differences in the leading $O$-notation constants) as large as the core, though the clique number and degeneracy differ significantly. Our upper bound on the degeneracy provides a constant factor approximation algorithm for the graph colouring problem. The approximation ratio ranges from $2/\sqrt{3}$ to $4/3$ and depends on the model parameter $\alpha$. This raises several open questions and future research directions.
\begin{itemize}
    \item Is the chromatic number bounded away from the degeneracy, the clique number, or both?
    \item Can HRGs be coloured optimally in polynomial time or is it NP-complete? 
    \item What are the asymptotics of $\clique/\coresize$? Is the clique number a constant factor larger than the core and has similar behaviour as the degeneracy? 
 \end{itemize}
There are further directions of research such as determining other differences between HRGs and GIRGs or designing colouring algorithms for HRGs in various models of computation. 

\bibliography{literature}
\appendix
\section{Appendix}\label{appendix:start}
\subsection{Bounds on $\theta_R(\cdot,\cdot)$}
Before we give bounds on the measure for of the area inner-ball \cref{cor:inner-neighbourhood-equal} in \cref{sec:app-area}, we provide some refined bounds on $\theta_R(r,q)$ for several pair of radii with $r + q \geq R + \Delta$ for specific $\Delta \in \Theta(1)$. We use the following inequality for the $\arccos$.

\begin{lemma}\label{rmk:arccos-bounds}
Let $0 < \varepsilon \leq x \leq z \leq 2$. Then
\begin{align*}
   \arccos(1-\varepsilon)\sqrt{x/\eps} \leq \arccos(1-x) \leq \arccos(1-z)\sqrt{x/z}.
\end{align*}
\end{lemma}

\begin{proof}
The statement follows by observing that $\arccos(1-x)x^{-1/2}$ is increasing in $x$ for domain $(0,2]$. 

\end{proof}

Recall that $\theta_R(r,q)$ is the angle between points with radii $p$ and $q$ such that the distance between them is $R$. We apply \cref{rmk:arccos-bounds} to $\theta_R(r,q)$ to obtain the following.

\begin{lemma}\label{lem:theta-bounds}
    For $u, v \in \disk$ and $\Delta \in \Theta(1)$, let $r(u) = r \geq q = r(v)$ and $r + q = R + \Delta$. Then there exists a constant $\delta > 0$, such that
    \begin{align*}
        (2+\delta)\sqrt{e^{R-r-q}} < \theta_R(r,q) \leq 
        \begin{cases}
\pi \sqrt{e^{R-r-q}} &\text{ for } \Delta \geq 0,\\
\frac{4\pi}{3\sqrt{3}}\sqrt{e^{R-r-q}} &\text{ for } \Delta \geq \log(4/3), \\
\frac{\pi}{\sqrt{2}}\sqrt{e^{R-r-q}} &\text{ for } \Delta \geq \log(2), \\
\frac{2\pi}{3}\sqrt{e^{R-r-q}} &\text{ for } \Delta \geq 2\log(2).
\end{cases}
\end{align*}
\end{lemma}

\begin{proof}
We adapt the proof of \cite[Lemma 4]{bfks-unitdisk-23}. Let us start with the upper bounds. We take advantage of the fact that $\arccos$ is a decreasing function. This means that we can upper bound $\theta_R(r,q)$:
\begin{align*}
    \theta_R(r,q) = \arccos\left(\frac{\cosh(r)\cosh(q) - \cosh(R)}{\sinh(r)\sinh(q)}\right) \leq \arccos\left(1 - 2e^{R-r-q}\right).
\end{align*}
Since $2e^{R-r-q} \in \left(0, 2e^{-\Delta}\right]$ we can apply \cref{rmk:arccos-bounds}. Evaluating the $\arccos$ for $r+q=R+ \Delta$ for the domains of our interest we obtain
\begin{align*}
    \arccos\left(1-2e^{R-r-q}\right) \leq
    \begin{cases}
        \pi  &\text{ for } \Delta \geq 0,\\
        \frac{2\pi}{3} &\text{ for } \Delta \geq \log(4/3), \\
        \frac{\pi}{2} &\text{ for } \Delta \geq \log(2), \\
        \frac{\pi}{3} &\text{ for } \Delta \geq 2\log(2),
    \end{cases}
\end{align*}

and the respective upper bound for $\theta_R(r,q)$ and each domain follows.

We now show a lower bound for $\theta_R(r,q)$. Since $\arccos$ is monotonic decreasing we upper bound its inner term to lower bound $\theta_R(r,q)$. By the identity $\cosh(r)\cosh(q) = \sinh(r)\sinh(q) - \cosh(r-q)$, we get
\begin{align*}
     \theta_R(r,q) &= \arccos\left(\frac{\cosh(r)\cosh(q) - \cosh(R)}{\sinh(r)\sinh(q)}\right)\\
     &= \arccos\left(1 - \frac{\cosh(R) - \cosh(r-q)}{\sinh(r)\sinh(q)}\right)\\
     &\geq \arccos\left(1-2e^{R-r-q} + e^{-R-r-q} - \left(e^{-2q} + e^{-2r}\right)\right).  
\end{align*}
Note that, by our hypothesis that $q \leq r \leq R/2 + \mathcal{O}(1)$, there exists an $\varepsilon > 0$, such that $2e^{R-x-y} + e^{-R-x-y} - \left(e^{-2x} + e^{-2r}\right) \in \left[\varepsilon, 2\right)$. Thus, we can and we will apply \cref{rmk:arccos-bounds}. We note that, for all $\varepsilon > 0$ there exists a constant $\delta' > 0$ such that $\arccos(1-\varepsilon)\sqrt{\frac{2}{\varepsilon}} = 2+\delta'$. Hence there exists another constant $\delta$ such that $0 < \delta < \delta'$, for which we get
\begin{align*}
    &\arccos\left(1-2e^{R-r-q} + e^{-R-r-q} - \left(e^{-2q} + e^{-2r}\right)\right)\\ 
    &\geq (2+\delta')\sqrt{e^{R-r-q} + \frac{e^{-R-r-q} - \left(e^{-2q} + e^{-2r}\right)}{2}}
    \geq (2+\delta)\sqrt{e^{R-r-q}},
\end{align*}
since $\left(e^{-2q} + e^{-2r}\right) \in \mathcal{O}(1/n)$ by $r \geq x \geq R/2 = \log(n)$ + C/2.
\end{proof}

We also consider vertices that have a sub-constant distance to the core in \cref{prop:larger-cliques}.

\forbiddenAngle*
\begin{proof}

Since $\arccos$ is a monotonically decreasing function and we aim to lower bound $\theta_R(r,x)$, we apply the identity $\cosh(r)\cosh(x) = \sinh(r)\sinh(x) + \cosh(r-x)$ and get
\begin{align*}
    \theta_R(r,x) &\geq \arccos\left(1 - 2\left(e^{R-r-x} + e^{-R-r-x} - \left(e^{-2r} +e^{-2x}\right)\right)\right).
\end{align*}
Next, we apply $r+x = R + \xi_k$ and $R = 2\log(n) + C$ which leads to
\begin{align*}
    \theta_R(r,x) 
    &\geq \arccos\left(1 - 2\left(e^{-\xi_k} + e^{-(4\log(n) + 2C +\xi_k)} - \left(e^{-2r} +e^{-2x}\right)\right)\right)\\
    &\geq \arccos\left(1 - 2\left(e^{-\xi_k} + e^{-(4\log(n) + 2C +\xi_k)} - 2e^{-(2\log(n) + C)}\right)\right),
\end{align*}
where we applied $r > x = R/2$ in the second line.
Using again the fact that $\arccos$ is a monotonically decreasing function we have
\begin{align*}
     \theta_R(r,x) &\geq \arccos\left(1 - 2\left(e^{-\xi_k} - 2e^{-(2\log(n) + C)}\right)\right).
\end{align*}
Note that for $\xi_k \in o(1)$ and $n$ approaching infinity the term $2\left(e^{-\xi_k} - 2e^{-(2\log(n) + C)}\right)$ approaches $2$ from below. Hence, a Taylor expansion around $n=\infty$ now yields
\begin{align*}
    \theta_R(r,x) &\geq \pi - 2\sqrt{\log^k(n) n^{\alpha-1}}.
\end{align*}
as required.
\end{proof}

\subsection{Bounds on the inner-neighbourhood}\label{sec:app-area}

\cref{lem:theta-bounds} allows us to bound the probability of a vertex $v$ belonging to the inner-neighbourhood area of a vertex $u$. We remark that these bounds are sharper than the ones obtained in \cite[Lemma 3]{bfm-cliques-17} for points with radius $R/2 + \Theta(1)$. In contrast, for points with radius $R/2 + \omega(1)$ the bound in \cite[Lemma 3]{bfm-cliques-17} is superior.

\begin{lemma}\label{lem:inner-neighoubrhood-upper}
Let $\Delta \in [0, R/2)$ and consider a point $u$ with radial coordinate $r = R/2 + \Delta$. Then, for the inner-ball $\mathcal{I}(r)$,
\begin{align*}
    \mu(\mathcal{I}(r)) \leq \left(1+ \Theta\left(e^{-\alpha R}\right)\right)\frac{\alpha e^{-\alpha r}}{\alpha - 1/2}\left(\gamma e^{\frac{1}{2}(2\alpha - 1)(2r-R)} - \eta\right),
\end{align*}
where
\begin{align*}
     \gamma, \eta =
\begin{cases}
1, \frac{1}{2\alpha} &\text{ for } \Delta \geq 0,\\
\frac{4}{3\sqrt{3}},\frac{1}{2\alpha} - \left(1-\frac{4}{3\sqrt{3}}\right)\left(\frac{4}{3}\right)^{(\alpha - 1/2)}  &\text{ for }  \Delta \geq \log(\sqrt{4/3}),\\
\frac{1}{\sqrt{2}},\frac{1}{2\alpha} - \left(1-\frac{4}{3\sqrt{3}}\right)\left(\frac{4}{3}\right)^{(\alpha - 1/2)}  -\left(\frac{4}{3\sqrt{3}} - \frac{1}{\sqrt{2}}\right)2^{(\alpha - 1/2)} &\text{ for } \Delta \geq \log(\sqrt{2}), \\
\frac{2}{3},\frac{1}{2\alpha} - \left(1-\frac{4}{3\sqrt{3}}\right)\left(\frac{4}{3}\right)^{(\alpha - 1/2)}  -\left(\frac{4}{3\sqrt{3}} - \frac{1}{\sqrt{2}}\right)2^{(\alpha - 1/2)} -\left(\frac{1}{\sqrt{2}} -\frac{2}{3}\right)2^{(2\alpha -1)} &\text{ for }  \Delta \geq \log(2). \\
\end{cases} 
 \end{align*}
\end{lemma}

\begin{proof}
The proof is an adaptation of \cite[Lemma 3.2]{gpp-hrg-12}. We calculate the measure of the inner-ball by integrating over the desired area. Since a vertex $u$ with radius $r$ has $\dist(u,v) \leq R$ if $r(v) \leq R - r$ we have for the area of an inner-ball
\begin{align*}
    \mu(\mathcal{I}(r)) &= \mu(\mathcal{B}_u(R) \cap \mathcal{B}_0(r)) = \mu(\mathcal{B}_0(R-r)) + 2 \int_{R-r}^r \int_0^{\theta_R(x,r)}\rho(x) \dd\theta dx.
\end{align*}
By \cref{lem:theta-bounds} we have  $\theta_R(r,q) \leq \pi \sqrt{e^{(R-r-x)}}$ for $x+r\geq R$. Thus, for every $x \in [R-r, r]$, there exists a $\tau \leq \pi$ such that 
\begin{align*}
    &\mu(\mathcal{B}_0(R-r)) + 2 \int_{R-r}^r \tau e^{\frac{R-r-x}{2}}\frac{\alpha\sinh(\alpha x)}{2\pi(\cosh(\alpha R) - 1)} dx\\
    &= \mu(\mathcal{B}_0(R-r)) + \frac{ \alpha e^{(R-r)/2}}{\pi(\cosh(\alpha R) - 1)}\int_{R-r}^{r}\tau e^{-x/2}\sinh(\alpha x) dx.
\end{align*}

Notice that we can split the integral of the domain $[R-r, r]$ into further sub domains such that for domain $[s, t]$ we get $s + r \in R + \Theta(1)$. This allows to apply different bounds for $\theta_R(s,r)$ as shown in \cref{lem:theta-bounds}, giving a more refined upper bounds on $\tau$. Thus, we apply \cref{lem:theta-bounds} to obtain  
\begin{align}
    \mu(\mathcal{I}(r))&\leq \mu(\mathcal{B}_0(R-r))\label{eq:innerdegree-0}\\ 
    &\quad+ \indicator_{\{r \geq R/2\}}\frac{\alpha e^{(R-r)/2}}{(\cosh(\alpha R) - 1)}\int_{R-r}^{\min(r, R-r+\log(4/3))}e^{-x/2}\sinh(\alpha x) \dd x\label{eq:innerdegree-1}\\
    &\quad+ \indicator_{\{r \geq R/2 + \log(\sqrt{4/3})\}}\frac{{4} \alpha e^{(R-r)/2}}{{3\sqrt{3}}(\cosh(\alpha R) - 1)}\int_{R-r+\log(4/3)}^{\min(r, {R-r+\log(2)})}e^{-x/2}\sinh(\alpha x)\dd x\label{eq:innerdegree-2}\\
    &\quad+ \indicator_{\{r\geq R/2 + \log(\sqrt{2})\}}\frac{\alpha e^{(R-r)/2}}{\sqrt{2}(\cosh(\alpha R) - 1)}\int_{R-r+\log(2)}^{{\min(r,R-r+2\log(2))}}e^{-x/2}\sinh(\alpha x)\dd x\label{eq:innerdegree-3}\\
    &\quad+ \indicator_{\{r\geq R/2 + 2\log(\sqrt{2})\}}\frac{2 \alpha e^{(R-r)/2}}{3(\cosh(\alpha R) - 1)}\int_{R-r+2\log(2)}^{r}e^{-x/2}\sinh(\alpha x)\dd x\label{eq:innerdegree-4}.
\end{align}
Before we consider each desired case for $\Delta$ separately, we observe that each integral is of the form $\int_{s}^{t}e^{-x/2}\sinh(\alpha x) \dd x$. Doing the calculations yield
\begin{align*}
    \int_{s}^{t}e^{-x/2}\sinh(\alpha x) &= \left[\frac{2}{4\alpha^2 - 1}e^{-x/2}(2\alpha \cosh(\alpha x ) + \sinh(\alpha x))\right]_s^t\\
    &= \frac{2}{4\alpha^2 -1}\left(e^{-t/2}(2\alpha \cosh(\alpha t) + \sinh(\alpha t))\right)\\
    &-\frac{2}{4\alpha^2 -1}\left(e^{-s/2}(2\alpha \cosh(\alpha s) + \sinh(\alpha s)\right).
\end{align*}

We note that each for $y,z \in \mathbb{Q}$ and  $k,m \in \mathbb{N}$, we can write $s$ as $R-r + k\log(y)$ and $t$ either as $k'\log(z)$ or $r$. First we plug in  $s = R-r +k\log(y)$ and $t = R-r + k'\log(z)$ and get 

\begin{align*}
    (4\alpha^2 - 1)&e^{(R-r)/2}\int_{R-r+k\log(y)}^{R-r+k'\log(z)}e^{-x/2}\sinh(\alpha x)\\
    &= 2 (z^{-k'/2}(2\alpha \cosh(\alpha(R-r+k'\log(z))) + \sinh(\alpha(R-r+k'\log(z)))) \\
    &\quad- 2 (y^{-k/2}(2\alpha \cosh(\alpha(R-r+k\log(y))) + \sinh(\alpha(R-r+k\log(y))))\\
    &= 2(z^{-k'/2}(z^{k'\alpha}(\alpha + 1/2)e^{\alpha (R-r)} + z^{-k'\alpha}(\alpha - 1/2)e^{-\alpha(R-r)})) \\
    &\quad- 2(y^{-k/2}(y^{k\alpha}(\alpha + 1/2)e^{\alpha (R-r)} + y^{-k\alpha}(\alpha - 1/2)e^{-\alpha(R-r)}) \\
    &= (2\alpha +1)\left(z^{k'(\alpha - 1/2)} e^{\alpha(R-r)} - y^{k(\alpha - 1/2)} e^{\alpha(R-r)} \right)\\
    &\quad- (2\alpha - 1)\left(y^{-k(1/2 +\alpha)}e^{-\alpha(R - r)} - z^{-k'(\alpha+1/2)}e^{-\alpha(R - r)} \right)\\
    &\le (2\alpha +1)\left(z^{k'(\alpha - 1/2)} - y^{k(\alpha - 1/2)} \right)e^{\alpha(R-r)} \qquad\tag{since $\alpha>1/2$}
\end{align*}
Multiplying both sides by $\tau\alpha/\pi(4\alpha^2 - 1)(\cosh(\alpha R) - 1)$, and noting that $(\cosh(\alpha R) - 1)^{-1} = 2\exp({-\alpha R})\left(1+ \Theta\left(e^{-\alpha R}\right)\right)$, yields
\begin{align}\label{eq:integral-s-t}
\begin{split}
    \frac{\tau\alpha e^{(R-r)/2}}{\pi(\cosh(\alpha R) - 1)}&\int_{R-r+k\log(y)}^{R-r+k'\log(z)}e^{-x/2}\sinh(\alpha x)\\
    &\le \left(1+\Theta\left(e^{-\alpha R}\right)\right)\frac{\alpha\tau e^{-\alpha r}}{\pi(\alpha -1/2)}\left(z^{k'(\alpha - 1/2)} - y^{k(\alpha - 1/2)} \right).
\end{split}
\end{align}
By similar calculations it is revealed for $t = r$
\begin{align}\label{eq:case-r-upper}
\begin{split}
     \frac{\tau\alpha e^{(R-r)/2}}{\pi(\cosh(\alpha R) - 1)}&\int_{R-r+k\log(y)}^{r}e^{-x/2}\sinh(\alpha x)\\ &{\le} \left(1+\Theta\left(e^{-\alpha R}\right)\right){\frac{\alpha\tau e^{-\alpha r}}{\pi(\alpha -1/2)}\left(e^{\frac{1}{2}(2\alpha - 1)(2r-R)} - y^{k(\alpha - 1/2)}\right)}.
\end{split}
\end{align}

  We proceed by taking care of each line of \cref{eq:innerdegree-0,eq:innerdegree-1,eq:innerdegree-2,eq:innerdegree-3,eq:innerdegree-4} separately. We start by considering \cref{eq:innerdegree-1}. Recall that $\tau \leq \pi$ for $r+s\geq R$. Using \eqref{eq:case-r-upper}  and \eqref{eq:integral-s-t} for the cases $\Delta < \log(\sqrt{4/3})$ and $\Delta \geq \log(\sqrt{4/3})$ respectively and using that for $r \leq R/2$ we get $\indicator_{\{r\geq R/2\}} = 0$ gives us
  \begin{align*}
      \cref{eq:innerdegree-1} \leq \left(1+\Theta\left(e^{-\alpha R}\right)\right)\frac{\alpha e^{-\alpha r}}{(\alpha -1/2)}\begin{cases}
          0 & \text{for } r \leq R/2,\\
          {\left(e^{\frac{1}{2}(2\alpha - 1)(2r-R)} - 1\right)} & \text{for } 0<\Delta<\log(\sqrt{4/3}),\\
          \left(\left(\frac{4}{3}\right)^{(\alpha - 1/2)} - 1\right) & \text{for } \Delta \geq \log(\sqrt{4/3}).
      \end{cases}
  \end{align*}
Moving on to \cref{eq:innerdegree-2,eq:innerdegree-3} , we get in similar fashion
\sam{missing some operators?}
\begin{align*}
     \cref{eq:innerdegree-2} \leq &\left(1+\Theta\left(e^{-\alpha R}\right)\right)\\
     &\frac{4\alpha e^{-\alpha r}}{3\sqrt{3}(\alpha -1/2)}\begin{cases}
     0 & \text{for } \Delta \leq \log(\sqrt{4/3}),\\
          {\left(e^{\frac{1}{2}(2\alpha - 1)(2r-R)} - \left(\frac{4}{3}\right)^{(\alpha - 1/2)}\right)} & \text{for } \log(\sqrt{4/3})<\Delta<\log(\sqrt{2}),\\
          \left(2^{(\alpha - 1/2)} - \left(\frac{4}{3}\right)^{(\alpha - 1/2)}\right) & \text{for } \Delta \geq \log(\sqrt{2}),
     \end{cases}
\end{align*}
and
\begin{align*}
     \cref{eq:innerdegree-3} \leq &\left(1+\Theta\left(e^{-\alpha R}\right)\right)\\
     &\frac{\alpha e^{-\alpha r}}{\sqrt{2}(\alpha -1/2)}\begin{cases}
     0 & \text{for } \Delta \leq \log(\sqrt{2}),\\
          {\left(e^{\frac{1}{2}(2\alpha - 1)(2r-R)} - 2^{(\alpha - 1/2)}\right)} & \text{for } \log(\sqrt{2})<\Delta<2\log(\sqrt{2}),\\
          \left(2^{(2\alpha - 1)} - 2^{(\alpha - 1/2)}\right) & \text{for } \Delta \geq 2\log(\sqrt{2}).
     \end{cases}
\end{align*}
For the last equation \eqref{eq:innerdegree-4} we only have to consider two cases and get
\begin{align*}
     \cref{eq:innerdegree-4} \leq &\left(1+\Theta\left(e^{-\alpha R}\right)\right)\\
     &\frac{2\alpha e^{-\alpha r}}{3(\alpha -1/2)}\begin{cases}
     0 & \text{for } \Delta \leq 2\log(\sqrt{2}),\\
          {\left(e^{\frac{1}{2}(2\alpha - 1)(2r-R)} - 2^{(2\alpha - 1)}\right)} & \text{else.} 
     \end{cases}
\end{align*}
Before putting everything together we use \cref{eq:measure-of-origin-ball} and get
\begin{align*}
    \cref{eq:innerdegree-0} \leq \left(1+\Theta\left(e^{-\alpha R}\right)\right)e^{-\alpha r}.
\end{align*}
It is left to consider each cases for our different choices of $\Delta$ individually and add up the necessary terms.

\textbf{Case 1}~[$0\leq\Delta<\log(\sqrt{4/3})$]: All but the first indicator variable is $0$, so we only have to consider the sum of \cref{eq:innerdegree-0} and \cref{eq:innerdegree-1}. This yields
\begin{align*}
    \mu(\mathcal{I}(r))&\leq \left(1+\Theta\left(e^{-\alpha R}\right)\right)e^{-\alpha r}\Big(1 + \frac{\alpha }{\alpha -1/2} {\left(e^{\frac{1}{2}(2\alpha - 1)(2r-R)} - 1\right)}\Big), 
\end{align*}
and we are done with the case after simple algebraic manipulation.

\textbf{Case 2}~[$\log(\sqrt{4/3})\leq\Delta<\log(\sqrt{2})$]:  In addition to \cref{eq:innerdegree-0} and \cref{eq:innerdegree-1}, also \cref{eq:innerdegree-2} is active now, and thus
\begin{align*}
    \mu(\mathcal{I}(r))&\leq \left(1+\Theta\left(e^{-\alpha R}\right)\right)e^{-\alpha r}\left(1 + \frac{\alpha }{\alpha -1/2} {\left(\left(\frac{4}{3}\right)^{(\alpha - 1/2)} - 1\right)}\right)\\ 
    &\quad+ \left(1+\Theta\left(e^{-\alpha R}\right)\right)\left(\frac{4\alpha e^{-\alpha r}}{3\sqrt{3}(\alpha -1/2)}{\left(e^{\frac{1}{2}(2\alpha - 1)(2r-R)} - \left(\frac{4}{3}\right)^{(\alpha - 1/2)}\right)}\right)\\
    &=\left(1+\Theta\left(e^{-\alpha R}\right)\right)\frac{\alpha e^{-\alpha r}}{\alpha - 1/2}\left(\left(1-\frac{4}{3\sqrt{3}}\right)\left(\frac{4}{3}\right)^{(\alpha - 1/2)} - \frac{1}{2\alpha} \right)\\ 
    &\quad+ \left(1+\Theta\left(e^{-\alpha R}\right)\right)\frac{4\alpha e^{-\alpha r}}{3\sqrt{3}(\alpha -1/2)}{e^{\frac{1}{2}(2\alpha - 1)(2r-R)} }.
\end{align*}

\textbf{Case 3}~[$\log(\sqrt{2})\leq\Delta<2\log(2)$]: To the previous case we add \cref{eq:innerdegree-3} to obtain
\begin{align*}
    \mu(\mathcal{I}(r))&\leq \left(1+\Theta\left(e^{-\alpha R}\right)\right)e^{-\alpha r}\left(1 + \frac{\alpha }{\alpha -1/2} {\left(\left(\frac{4}{3}\right)^{(\alpha - 1/2)} - 1\right)}\right)\\
     &\quad+ \left(1+\Theta\left(e^{-\alpha R}\right)\right)e^{-\alpha r}\left(\frac{4\alpha}{3\sqrt{3}(\alpha -1/2)}\left(2^{(\alpha -1/2)} - \left(\frac{4}{3}\right)^{(\alpha - 1/2)} \right)\right)\\
    &\quad+ \left(1+\Theta\left(e^{-\alpha R}\right)\right)\left(\frac{\alpha e^{-\alpha r}}{\sqrt{2}(\alpha -1/2)}{\left(e^{\frac{1}{2}(2\alpha - 1)(2r-R)} - 2^{(\alpha - 1/2)}\right)}\right)\\
    &=\left(1+\Theta\left(e^{-\alpha R}\right)\right)\frac{\alpha e^{-\alpha r}}{\alpha - 1/2}\\
    &\quad\cdot\left(\left(\frac{4}{3\sqrt{3}} - \frac{1}{\sqrt{2}}\right)2^{(\alpha - 1/2)} + \left(1-\frac{4}{3\sqrt{3}}\right)\left(\frac{4}{3}\right)^{(\alpha - 1/2)} - \frac{1}{2\alpha} \right)\\ 
    &\quad+ \left(1+\Theta\left(e^{-\alpha R}\right)\right)\frac{\alpha e^{-\alpha r}}{\sqrt{2}(\alpha -1/2)}{e^{\frac{1}{2}(2\alpha - 1)(2r-R)} }.
\end{align*}

\textbf{Case 4}~[$\Delta\geq2\log({2})$]: Finally, we consider the last case and compute in conjunction with \cref{eq:innerdegree-4}
\begin{align*}
     \mu(\mathcal{I}(r))&\leq \left(1+\Theta\left(e^{-\alpha R}\right)\right)e^{-\alpha r}\left(1 + \frac{\alpha }{\alpha -1/2} {\left(\left(\frac{4}{3}\right)^{(\alpha - 1/2)} - 1\right)}\right)\\
     &\quad+ \left(1+\Theta\left(e^{-\alpha R}\right)\right)e^{-\alpha r}\left(\frac{4\alpha}{3\sqrt{3}(\alpha -1/2)}\left(2^{(\alpha -1/2)} - \left(\frac{4}{3}\right)^{(\alpha - 1/2)} \right)\right)\\
    &\quad+\left(1+ \Theta\left(e^{-\alpha R}\right)\right)e^{-\alpha r}\left(\frac{\alpha}{\sqrt{2}(\alpha -1/2)}\left(2^{(2\alpha-1)} - 2^{(\alpha-1/2)}\right)\right)\\
    &\quad+ \left(1+ \Theta\left(e^{-\alpha R}\right)\right)e^{-\alpha r}\left(\frac{2\alpha}{3(\alpha -1/2)}\left(e^{\frac{1}{2}(2\alpha - 1)(2r-R)}  - 2^{(2\alpha -1)} \right)\right)\\
    &=\left(1+\Theta\left(e^{-\alpha R}\right)\right)\frac{\alpha e^{-\alpha r}}{\alpha - 1/2}\\
    &\quad\cdot\Bigg(\left(\frac{1}{\sqrt{2}} -\frac{2}{3}\right)2^{(2\alpha -1)} + \left(\frac{4}{3\sqrt{3}} - \frac{1}{\sqrt{2}}\right)2^{(\alpha - 1/2)}\\
    &\qquad+ \left(1-\frac{4}{3\sqrt{3}}\right)\left(\frac{4}{3}\right)^{(\alpha - 1/2)} - \frac{1}{2\alpha} \Bigg)\\ 
    &\quad+ \left(1+\Theta\left(e^{-\alpha R}\right)\right)\frac{2\alpha e^{-\alpha r}}{3(\alpha -1/2)}{e^{\frac{1}{2}(2\alpha - 1)(2r-R)}}.
\end{align*}
\end{proof}

We complement the upper bound with a lower bound for the measure of an inner-ball. Combining this with \cref{lem:deg-lower-inner} allows us to lower bound the degeneracy.

\begin{lemma}\label{lem:inner-neighoubrhood-lower}
 Let $R/2 \leq r \leq R$. 
 Then, for the inner-ball $\mathcal{I}(r)$,
    \begin{align*}
        \mu(\mathcal{I}(r)) \geq \left(1+ \Theta\left(e^{-\alpha R}\right)\right)\frac{\alpha e^{-\alpha r}}{(\alpha - 1/2)}\left(\frac{2}{\pi}e^{\frac{1}{2}(2\alpha - 1)(2r-R)} - \left(\frac{2}{\pi} - \frac{(\alpha - 1/2)}{\alpha}\right)\right).
    \end{align*}
\end{lemma}
\begin{proof}
    The calculations are similar but simpler to those for \cref{lem:inner-neighoubrhood-upper}. Recall that
    \begin{align*}
    \mu(\mathcal{I}(r)) &= \mu(\mathcal{B}_u(R) \cap \mathcal{B}_0(r)) = \mu(\mathcal{B}_0(R-r)) + 2 \int_{R-r}^r \theta_R(x,r)\rho(x) \dd x.
\end{align*}
This, in conjunction with \cref{lem:theta-bounds} yields for a constant $\delta > 0$
\begin{align*}
    \mu(\mathcal{I}(r)) &\geq\frac{(2+\delta)\alpha e^{(R-r)/2}}{\pi(\cosh(\alpha R) - 1)} \int_{R-r}^{r}e^{-x/2}\sinh(\alpha x)\\
        &= \frac{2(2+\delta) \alpha (e^{R/2-r}(2\alpha \cosh(\alpha r) + \sinh(\alpha r)))}{\pi(4\alpha^2 - 1)(\cosh(\alpha R) -1)} \\
         &- \frac{2(2+\delta) \alpha(2\alpha \cosh(\alpha(R-r)) + \sinh(\alpha(R-r)))}{\pi(4\alpha^2 - 1)(\cosh(\alpha R) -1)} \\
        &= \frac{2(2+\delta) \alpha(e^{R/2 - r}((\alpha + 1/2)e^{\alpha r} + (\alpha - 1/2)e^{-\alpha r}))}{\pi(4\alpha^2 - 1)(\cosh(\alpha R) -1)} \\
        &- \frac{2(2+\delta) \alpha((\alpha + 1/2) e^{\alpha (R-r)} + (\alpha - 1/2)e^{-\alpha(R-r)})}{\pi(4\alpha^2 - 1)(\cosh(\alpha R) -1)} \\
        &= \frac{\frac{\alpha(2+\delta)}{\pi(2\alpha -1)}\left(e^{R/2 - (1-\alpha)r} - e^{\alpha(R-r)} \right) - \frac{\alpha(2+\delta)}{\pi(2\alpha + 1)}\left(e^{-\alpha(R - r)} - e^{-(1 +\alpha)r + R/2} \right)}{(\cosh(\alpha R) -1)}\\
        &\stackrel{\delta > 0}{\geq} \frac{\frac{2\alpha}{\pi(2\alpha -1)}\left(e^{R/2 - (1-\alpha)r} - e^{\alpha(R-r)} \right)}{(\cosh(\alpha R) -1)},
    \end{align*}
where the last line follows since $\frac{\delta\alpha}{\pi(2\alpha -1)}\left(e^{R/2 - (1-\alpha)r} - e^{\alpha(R-r)}\right) \in \Theta(1)$ by the hypothesis that $R/2 < r \leq R/2 + \mathcal{O}(1)$ and $\delta > 0$ constant, in contrast to the vanishing term $\frac{\alpha(2+\delta)}{\pi(2\alpha + 1)}\left(e^{-\alpha(R - r)} - e^{-(1 +\alpha)r + R/2}\right) \in o(1)$.

Putting this together with the fact that $\mu(\mathcal{B}_0(R-r)) = \int_{0}^{R -r} \frac{\alpha \sinh{(\alpha x)}}{\cosh{(\alpha R)} - 1}\dd x = \frac{\cosh(\alpha(R-r) - 1}{\cosh{(\alpha R)} - 1}$ (see e.g. \cref{eq:measure-of-origin-ball}) it follows that
\begin{align*}
    \mu(\mathcal{I}(r)) &= \mu(\mathcal{B}_u(R) \cap \mathcal{B}_0(r)) = \mu(\mathcal{B}_0(R-r)) + 2 \int_{R-r}^r \theta_R(x,r)\rho(x) \dd x\\
    &\geq \left(1+ \Theta\left(e^{-\alpha R}\right)\right)\left(e^{-\alpha r} + \frac{2\alpha}{\pi(\alpha -1/2)}\left(e^{-(\alpha-1/2)R - (1-\alpha)r} - e^{-\alpha r} \right)\right)\\
    &=\left(1+ \Theta\left(e^{-\alpha R}\right)\right)\frac{\alpha e^{-\alpha r}}{(\alpha - 1/2)}\left(\frac{2}{\pi}e^{\frac{1}{2}(2\alpha - 1)(2r-R)} - \left(\frac{2}{\pi} - \frac{(\alpha - 1/2)}{\alpha}\right)\right),
\end{align*}
where we used $\frac{1}{\cosh{(\alpha R)} - 1} = 2\exp({-\alpha R})\left(1+ \Theta\left(e^{-\alpha R}\right)\right)$. This concludes the proof.
\end{proof}

\subsection{Bounds on the degeneracy}

\degeneracyUpper*
\begin{proof}
We shall apply \cref{obs:deg-upper-inner} to upper bound the degeneracy by the maximal inner-degree. We accomplish this via the neighbourhood of a vertex that is formed within the inner-ball. Recall that \cref{lem:radius-maximal-innerdegree} gives us the radius $r^*$ by which the inner-neighbourhood $\mathcal{I}(r)$ is maximal. Moreover, by \cref{lem:inner-neighoubrhood-upper} we have an upper-bound on the inner-neighbourhood. Thus we upper bound the  measure of any inner-ball by plugging $r^*$ into the upper bound of $\mu(\mathcal{I}(r))$ to see that
\begin{align}\label{eq:upper-maximal-inner}
\begin{split}
    \mu(\mathcal{I}(r^*)) &\leq \left(1+ \Theta\left(e^{-\alpha R}\right)\right)\frac{\left(\frac{\gamma(1-\alpha)}{\alpha\eta}\right)^{\frac{\alpha}{2\alpha -1}}\alpha e^{-\alpha R/2}}{\alpha - 1/2}\left( \frac{\alpha\eta}{(1-\alpha)} - \eta\right)\\
    &= \left(1+ \Theta\left(e^{-\alpha R}\right)\right)\frac{2\alpha(\gamma(\frac{1}{\alpha}-1))^{\frac{\alpha}{2\alpha-1}}\eta^{\frac{\alpha-1}{2\alpha-1}}}{1-\alpha}e^{-\alpha R/2}.
\end{split}
\end{align}
We proceed by considering two different cases for $\alpha$, namely $1/2 < \alpha < 9/10$ and $9/10 \leq \alpha < 1$, and show that the largest inner-degree is at most $(4/3)^{\alpha}n^{1-\alpha}e^{-\alpha C/2} = ((4/3)^{\alpha} + o(1))\coresize$ \whp in both cases. By \cref{obs:deg-upper-inner} this proves the statement.

\noindent\textbf{Case 1} [$1/2 < \alpha < 9/10$]: Recall that by \cref{lem:radius-maximal-innerdegree}, $r^*$ is monotonically increasing in $\alpha$. Then, letting $\alpha$ approach $1/2$ from above, and setting $\gamma = \frac{1}{\sqrt{2}}$ and $\eta = \frac{1}{2\alpha} - (1-\frac{4}{3\sqrt{3}})(\frac{4}{3})^{(\alpha - 1/2)} - (\frac{4}{3\sqrt{3}} - \frac{1}{\sqrt{2}})2^{(\alpha - 1/2)}$, we get 
$
    r^* \geq R/2 + 1/2 > R/2 + \frac{1}{2}\log(2).
$
This satisfies all necessary conditions for \cref{lem:inner-neighoubrhood-upper} when $\Delta \geq \log(\sqrt{2})$. Applying \eqref{eq:upper-maximal-inner} and noticing that it \sam{it?} is monotonically increasing in $\alpha$, 
\begin{align}\label{eq:ordinary light}
     \mu(\mathcal{I}(r^*)) < \left(\frac{4}{3}\right)^{\alpha}n^{-\alpha}e^{-\alpha C/2}
\end{align}
for $\alpha < 9/10$.
Taking the expectation and a Chernoff bound we have that \whp~the inner-neighbourhood of a vertex with radius $r^*$ is at most $((4/3)^{\alpha}+o(1))\coresize$. By a union bound this holds for any vertex with any radius $r$ since $\mu(\mathcal{I}(r)) \leq \mu(\mathcal{I}(r^*))$.

\noindent\textbf{Case 2} [$9/10 \leq \alpha < 1$]: Taking $\gamma = \frac{2}{3}$ and $\eta = \frac{1}{2\alpha} - (1-\frac{4}{3\sqrt{3}})(\frac{4}{3})^{(\alpha - 1/2)} - (\frac{4}{3\sqrt{3}} - \frac{1}{\sqrt{2}})2^{(\alpha - 1/2)} - (\frac{1}{\sqrt{2}} -\frac{2}{3})2^{(2\alpha -1)}$,
we have via \cref{lem:radius-maximal-innerdegree} that for $\alpha \geq 9/10$,
\begin{align*}
    r^* \geq R/2 + 3/4 > R/2 + \log(2).
\end{align*}
Then, for such $\gamma$, $\eta$ and $\alpha$ we again obtain
\eqref{eq:ordinary light}, but now for $9/10 \leq \alpha < 1$.

The proof for the upper bound is then finished by another combination of a Chernoff and union bound, similarly to the previous case.

We prove the lower bound in two steps. First, we derive a radius $r^*$, such that a point with radius $r^*$ entails a relatively large lower bound for the measure of its inner-ball. Then, in a second step, we show that, \whp, there exists a vertex with radial coordinate within small radial distance $\xi$ to $r^*$. This then implies that the inner-degree of this vertex to be close to that of a vertex with radius $r^*$. This results in a lower bound for the inner-degree of a hyperbolic random graph \whp By \cref{lem:deg-lower-inner}, a lower bound on the maximal inner-degree then directly translates into a lower bound for the degeneracy \whp\ as given in the statement concluding our proof.

Taking $\gamma = \frac{2}{\pi}$ and $\eta = \frac{2}{\pi} - \frac{(\alpha - 1/2)}{\alpha}$ from \cref{lem:inner-neighoubrhood-lower} and plugging into \cref{lem:radius-maximal-innerdegree}, we get that
\begin{align*}
    r^* = R/2 + \frac{\log\left(\frac{\alpha\eta}{\gamma(1-\alpha)}\right)}{2\alpha - 1} = R/2 + \frac{\log\left(\frac{4\alpha + \pi - 2\pi\alpha}{4(1-\alpha)}\right)}{2\alpha - 1} \in R/2 + \Theta(1).
\end{align*}

We now show that there exists, \whp, a vertex with radius within the interval $r^*$ and $r^* + \xi$ and lower bound $\mu(\mathcal{I}(r^* + \xi))$. The result then follows by considering the expected number vertices within $\mathcal{I}(r^* + \xi)$ and finally using a chernooff bound to achieve the desired \whp result. 

With hindsight we choose $\xi = \log\left(1+\frac{\log^2(n)}{n^{1-\alpha}}\right)/\alpha$ and we consider the set of points $\mathcal{A}:=\mathcal{B}_0(r^* + \xi)\setminus\mathcal{B}_0(r^*)$. For our choice of $\xi$, and using $r^* > R/2$, we then get
\begin{align*}
    \mu(\mathcal{A}) &= \int_{r^*}^{r^* + \xi} \frac{\sinh(\alpha x)}{\cosh(\alpha R) - 1} = \frac{\cosh(\alpha(r^* + \xi)) - \cosh(\alpha r^*)}{\cosh(\alpha R) - 1}\\
    &\geq (1+\Theta(e^{-\alpha R}))(e^{-\alpha(R/2 - \xi)} - e^{-\alpha R/2})\\
    &=  (1+\Theta(e^{-\alpha R}))(e^{-\alpha R/2}e^{\log\left(1+\frac{\log^2(n)}{n^{1-\alpha}}\right)} - e^{-\alpha R/2})\\
    &= (1+\Theta(e^{-\alpha R}))(n^{-\alpha}e^{-\alpha C/2}{\left(1+\frac{\log^2(n)}{n^{1-\alpha}}\right)} - n^{-\alpha}e^{-\alpha C/2})\\
&= (1+\Theta(e^{-\alpha R}))\left(e^{-\alpha C/2}{\frac{\log^2(n)}{n}}\right) \in \omega(\log(n)/n).
\end{align*}
Notice that the expected number of vertices in $\mathcal{A}$ is then $\omega(\log(n))$ so using a Chernoff bound there are $\omega(\log(n))$ vertices in $\mathcal{A}$ \whp It is left to lower bound the inner-degree of any vertex contained in $\mathcal{A}$.

Making use of \cref{lem:inner-neighoubrhood-lower} in conjunction with $r^* = R/2 + \frac{\log\left(\frac{4\alpha + \pi - 2\pi\alpha}{4(1-\alpha)}\right)}{2\alpha - 1}$ and $\xi = \log\left(1+\frac{\log^2(n)}{n^{1-\alpha}}\right) / \alpha$ we have
\begin{align*}
    \mu(\mathcal{I}(r^* + \xi)) &\geq \left(1+ \Theta\left(e^{-\alpha R}\right)\right)\frac{\alpha e^{-\alpha (r^* + \xi)}}{\alpha - 1/2}\left(\frac{2}{\pi} e^{\frac{1}{2}(2\alpha - 1)(2(r^* + \xi)-R)} - \left(\frac{2}{\pi} - \frac{(\alpha - 1/2)}{\alpha}\right)\right)\\
    &= (1 - o(1))\frac{2\alpha(\frac2\pi(\frac{1}{\alpha}-1))^{\frac{\alpha}{2\alpha-1}}\left(\frac{2}{\pi} - \frac{(\alpha - 1/2)}{\alpha}\right)^{\frac{\alpha-1}{2\alpha-1}}}{1-\alpha}e^{-\alpha R/2}.
\end{align*}
Since $R = 2\log(n) + C$, we note that for any vertex $a \in V \cap \mathcal{A}$, the expected inner-degree $\mathbb{E}[\Gamma(a)]$ is at least
\begin{align*}
    \mathbb{E}[\Gamma(a)] \geq n\mu(\mathcal{I}(r^* + \xi)) &\geq \frac{(4-o(1))}{\pi}\left(\frac{2\left(1-\alpha\right)}{\frac{\pi}{2}-\alpha\left(\pi-2\right)}\right)^{\frac{1-\alpha}{2\alpha-1}} n^{1-\alpha}e^{-\alpha C/2}.
\end{align*}
Observe that the expected number of vertices is lower bounded by $n^{\Omega(1)}$. Hence, a final application of a Chernoff bound shows that there exists a vertex with inner-degree at least $(1 -\frac{1}{\log(n)})\mathbb{E}[\Gamma(a)]$ \whp This 
gives the stated lower bound by applying \cref{lem:deg-lower-inner} and recalling $\coresize = (1\pm o(1)) n^{1-\alpha}e^{-\alpha C/2}$ \whp  
\end{proof}

\subsection{Approximation algorithm}

\approxAlgo*

\begin{proof}
    We construct a colouring using the degeneracy. We do so following Matula and Beck~\cite{Matula1983} and compute a \emph{smallest-last vertex ordering}. To bound the run time of the algorithm we apply \cite[Lemma 1]{Matula1983} that states a run time of $\mathcal{O}(|E| + |V|)$. The average degree for hyperbolic random graph is $\mathcal{O}(1)$ \whp (see \cite[Corollary 17]{hrg-spectral}). Hence $|E| \in \mathcal{O}(n)$ \whp and by definition $|V| = n$. We conclude that it requires at most $\mathcal{O}(n)$ time \whp~to compute the smallest-last vertex ordering of a hyperbolic random graph. 

    Having the smallest-last vertex ordering in hand we colour the graph in the following greedy fashion: let $v_1,\cdots,v_n$ be the smallest-last vertex ordering of the vertices. Iterating through the vertices in this order, let $i \in [n]$ and colour vertex $v_i$ at time step $i$ with the smallest colour available that previously has not been assigned to any neighbour of $v_i$. More formally, consider the induced subgraph $G[V_i]$ where $V_i = \{v_1, v_2, .. v_i\}$ and set the colour for $v_i$ to be the colour $f(v_i) = \min{(\{\mathbb{N} \setminus N_i\})}$ where $N_i = \{f(u) | u \in V_{i-1} \wedge \{u,v_i\} \in E(G)\}$. This concludes the description of the algorithm. 
    
    The correctness of the algorithm is immediate by observing that when we assign a colour to a vertex $v_i$, this colour is different from the colours of all previously coloured neighbours of $v_i$. Notice that, given the smallest-last vertex ordering $v_1,\cdots, v_n$, the colouring can be constructed in $\mathcal{O}(n)$ \whp To back up this claim, we recall that $|E| \in \mathcal{O}(n)$ \whp and that the algorithm iterates through $n$ vertices while it checks the colour of all neighbours. Thus, for each edge $\{u, v_i\} \in E$, the algorithm checks the colour assigned to $u$ at step $i$ and over all iterations these are $|E| \in \mathcal{O}(n)$ checks \whp Again, we conclude with a run time of $\mathcal{O}(|E| + |V|)$ which boils down to $\mathcal{O}(n)$ \whp This concludes the correctness and our claimed running time $\mathcal{O}(n)$.

    We finish the proof by showing the claimed approximation ratio of $(4/3)^{\alpha}$. To this end, we first consider the number of colours used by our algorithm. Recall that, by \cref{the:degeneracy-upper}, the degeneracy of a hyperbolic random graph is at most $\kappa(G) \leq (1+o(1))\left(\frac{4}{3}\right)^{\alpha}\coresize$ \whp By  \cite[Lemma 4]{Matula1983}, the number of colours used via the smallest-last vertex ordering is upper bounded by the degeneracy $\kappa(G) + 1$. Moreover, the minimal amount of colours for $G$ are given by its chromatic number $\chi(G)$. Thus, an approximation ratio is given by $\frac{\kappa(G) + 1}{\chi(G)} \leq \frac{\kappa(G) + 1}{\coresize}$ since $\chi(G) \geq \coresize$ by \cref{lem:lower-clique-number}. Since the upper bound on $\kappa(G)$ and the lower bound on $\omega(G)$ holds \whp, a union over their complementary events does not occur \whp The \coresize for $\degen/\coresize$ cancels out yielding the desired approximation factor~\whp 
\end{proof}

\subsection{Geometric inhomogeneous random graphs}
\girgInner*

\begin{proof}
The upper bound follows by ordering the vertices $V = 
v_1,\cdots,v_n$ in ascending order by their respective weights so that for $i,j \in [n]$ and $i < j$ it holds $w_{v_i} \leq w_{v_j}$. Since in this ordering the vertex $u^*$ with largest inner-degree has the largest degree among the vertices with larger index it follows $\degen \leq |\innerdeg(\ustar)|$.   

For the lower bound we show that for every ordering of the vertices $V$, there must be a vertex with at least $(1-o(1))|\innerdeg(\ustar)|$ neighbours of larger index \whp To this end we use that for a fixed $w$ and $w_v \geq w_u$ that
\begin{align}\label{eq:yes}
    \Prob{x \in N(u) \cap w_x \geq w} &= 2\int_{w}^{\infty}t(y,w_u)f(y) dy\notag\\
    &\leq 2\int_{w}^{\infty}t(y,w_v)f(y) dy = \Prob{x \in N(v) \cap w_x \geq w},
\end{align}
since $t(y,w_v) \geq t(y,w_u)$ for $w_v \geq w_v$.
As an additional ingredient we use
\begin{equation}\label{eq:no}
    |\innerdeg(u^*)| \geq \clique \geq \coresize \in n^{\Omega(1)} \text{ \whp}
\end{equation}

Let $w^* = w_{\ustar}$ and we consider the set $U = \{v \in V | w_v \geq w^*\}$ and let $X_v$ be a random variable where $X_v = |\{u \in N(v)| u \in U\}|$. We then have
\begin{align*}
    \E{X_v} \geq \E{|\innerdeg(\ustar)|} \in n^{\Omega(1)}. \tag{by \eqref{eq:yes} and \eqref{eq:no}}
\end{align*}
Fixing a vertex $v \in U$ it holds $X_v \geq (1-1/\log(n))\E{|\innerdeg(\ustar)|}$ \whp using a Chernoff bound. A union bound now yields that this holds for all vertices in $U$. Now, in any ordering of $V$, the vertex of $U$ with lowest index has at least $(1-1/\log(n))\E{|\innerdeg(\ustar)|}$ neighbours with larger larger index \whp 
\end{proof}

\girgGap*

\begin{proof}
Fix any $\alpha \in (1/2, 1)$ as given in the statement. Draw the hyperbolic random graph  $H \sim \mathcal{G}(n, \alpha, C)$ and the geometric inhomogeneous random graph $G \sim \mathcal{G}(n, 2\alpha + 1, \lambda)$. Using \cref{the:girg-degeneracy} for the degeneracy of the GIRG we obtain~\whp $$\kappa(G)/\coresize \geq (2 - o(1)) (2(1-\alpha))^{({1 - \alpha})/{(2\alpha - 1)}}.$$
In contrast, upper bounding the degeneracy of the HRG with \cref{the:degeneracy-upper} reveals~\whp
$$
\kappa(H)/\sigma(H) \leq ((4/3)^{\alpha} - o(1)).
$$
Combining the two inequalities, it now follows that $
 \kappa(G)/\coresize  -  \kappa(H)/\sigma(H)  
    $
    is bounded away from $0$~\whp, which is what we sought out to show.
\end{proof}
\end{document}

%% file: preamble.tex
\usepackage{amsmath,amssymb,amsthm}
\usepackage{enumerate}
\usepackage{graphicx}
\usepackage{algorithm}
\usepackage{algpseudocode}
\usepackage{dsfont}
\usepackage{mathtools}
\usepackage{tikz-cd}
\usetikzlibrary{decorations.pathmorphing}
  
\usepackage[size=tiny]{todonotes}
\usepackage{todonotes}
\usepackage{thmtools}
\usepackage{thm-restate}
\usepackage{breqn}
\usepackage{xspace}
\theoremstyle{plain}
\newtheorem*{theorem*}{Theorem}

\usepackage{blindtext}

\crefname{observation}{Observation}{Observations}

\newcommand{\E}[1]{\mathbb{E}\left[#1\right]}
\newcommand{\Exp}{\mathbb E}
\newcommand{\Prob}[1]{\mathbb{P}\left(#1\right)}
\renewcommand{\P}{\mathbb P}

\newcommand{\eps}{\varepsilon}
\newcommand{\whp}{w.e.h.p.\xspace}
\newcommand{\dd}{\,\mathrm d}
\newcommand{\ee}{\mathrm e}

\newcommand{\abs}[1]{|#1|}
\newcommand{\B}{\mathcal B}
\newcommand{\Hb}{\ensuremath{\mathbb H^2}\xspace}
\newcommand{\bH}{\Hb}
\newcommand{\disk}{\ensuremath{\mathcal D_R}\xspace}
\newcommand{\G}{\ensuremath{\mathcal{G}(n, \alpha, C)}\xspace}
\newcommand{\GIRG}{\ensuremath{\mathcal{G}(n, \beta, \lambda)}\xspace}

\newcommand{\given}{\;|\;}

\newcommand{\ustar}{\ensuremath{u^*}\xspace}
\newcommand{\inner}{inner-neighbourhood\xspace}
\newcommand{\indeg}{inner-degree\xspace}
\newcommand{\derive}[1]{\frac{\mathrm d}{\mathrm d #1}}

\newcommand{\clique}{\ensuremath{\omega(G)}\xspace}
\newcommand{\chrom}{\ensuremath{\chi(G)}\xspace}
\newcommand{\degen}{\ensuremath{\kappa(G)}\xspace}
\newcommand{\coresize}{\ensuremath{\sigma(G)}\xspace}

\DeclarePairedDelimiter\floor{\lfloor}{\rfloor}

\DeclareMathOperator{\dist}{d_h}
\DeclareMathOperator{\innerdeg}{\Gamma}

\DeclareMathOperator{\indicator}{\mathds{1}}
\DeclareMathOperator{\forbidden}{\mathcal{F}_u}
\DeclareMathOperator{\hypercycle}{\mathcal{H}}

\newcommand{\sam}[1]{\ifdraft\textcolor[rgb]{0.9,0.1,0.1}{Sam: #1}\fi}